\documentclass[runningheads]{llncs}
\usepackage{graphicx}
\usepackage{amsmath,amssymb,cleveref,caption,subcaption,tikz}
\usepackage{algpseudocode}
\usepackage[shortlabels]{enumitem}

\usepackage[ruled,vlined,linesnumbered]{algorithm2e}
\usepackage[misc]{ifsym}

\newcommand{\calA}{\mathcal{A}}

\newcommand{\calS}{\mathcal{S}}

\newcommand{\E}{\mathbb{E}}

\newcommand{\setcov}{\mathsf{setcov}}

\DeclareMathOperator{\opt}{OPT}

\newcommand{\copt}{C_{\text{OPT}}}

\allowdisplaybreaks

\begin{document}
\title{Justifying Groups in \\Multiwinner Approval Voting}
%
%
\author{Edith Elkind\inst{1} \and
Piotr Faliszewski\inst{2} \and
Ayumi Igarashi\inst{3} \and
Pasin Manurangsi\inst{4} \and \\
Ulrike Schmidt-Kraepelin\inst{5} \and
Warut Suksompong\inst{6}
}
\authorrunning{E. Elkind et al.}
%
\institute{University of Oxford, UK \\
\email{elkind@cs.ox.ac.uk}
\and
AGH University of Science and Technology, Poland \\
\email{faliszew@agh.edu.pl}
\and
National Institute of Informatics, Japan \\
\email{ayumi\_igarashi@nii.ac.jp}
\and
Google Research, USA \\
\email{pasin@google.com}  \and
TU Berlin, Germany \\
\email{u.schmidt-kraepelin@tu-berlin.de} \and
National University of Singapore, Singapore  \\
\email{warut@comp.nus.edu.sg}
}
\maketitle              
\begin{abstract}
Justified representation (JR) is a standard notion of representation in multiwinner approval voting.
Not only does a JR committee always exist, but previous work has also shown through experiments that the JR condition can typically be fulfilled by groups of fewer than $k$ candidates, where $k$ is the target size of the committee. 
In this paper, we study such groups---known as \emph{$n/k$-justifying groups}---both theoretically and empirically.
First, we show that under the impartial culture model, $n/k$-justifying groups of size less than $k/2$ are likely to exist, which implies that the number of JR committees is usually large.
We then present efficient approximation algorithms that compute a small $n/k$-justifying group for any given instance, and a polynomial-time exact algorithm when the instance admits a tree representation.
In addition, we demonstrate that small $n/k$-justifying groups can often be useful for obtaining a gender-balanced JR committee even though the problem is NP-hard.
\end{abstract}

\section{Introduction}
\label{sec:intro}

Country X needs to select a set of singers to represent it in an international song festival.
Not surprisingly, each member of the selection board has preferences over the singers, depending possibly on the singers' ability and style or on the type of songs that they perform.
How should the board aggregate the preferences of its members and decide on the group of singers to invite for the festival?

The problem of choosing a set of candidates based on the preferences of voters---be it singers for a song festival selected by the festival's board, researchers selected by the conference's program committee to give full talks, or places to include on the list of world heritage sites based on votes by Internet users---is formally studied under the name of \emph{multiwinner voting}~\cite{FSST17}.
In many applications, the voters' preferences are expressed in the form of approval ballots, wherein each voter either approves or disapproves each candidate; this is a simple yet expressive form of preference elicitation \cite{BF07,K10}.
When selecting a committee, an important consideration is that this committee adequately represents groups of voters who share similar preferences.
A natural notion of representation, which was proposed by Aziz et al.~\cite{ABC17} and has received significant interest since then, is \emph{justified representation (JR)}.
Specifically, if there are $n$ voters and the goal is to select $k$ candidates, a committee is said to satisfy JR if for any group of at least $n/k$ voters all of whom approve a common candidate, at least one of these voters approves some candidate in the committee.

A committee satisfying JR always exists for any voter preferences, and can be found by several voting procedures \cite{ABC17}.
In fact, Bredereck et al.~\cite{BFKN19} observed experimentally that when the preferences are generated according to a range of stochastic distributions, the number of JR committees is usually very high.
This observation led them to introduce the notion of an \emph{$n/k$-justifying group}, which is a group of candidates that already fulfills the JR requirement even though its size may be smaller than $k$.
Bredereck et al.~found that, in their experiments, small $n/k$-justifying groups (containing fewer than $k/2$ candidates) typically exist.
This finding helps explain why there are often numerous JR committees---indeed, to obtain a JR committee, one can start with a small $n/k$-justifying group and then extend it with arbitrary candidates.

The goal of our work is to conduct an extensive study of $n/k$-justifying groups, primarily from a theoretical perspective but also through experiments.
Additionally, we demonstrate that small $n/k$-justifying groups can be useful for obtaining JR committees with other desirable properties such as gender balance.

\subsection{Our Contribution}
\label{sec:contribution}

In \Cref{sec:general-guarantees}, we present results on $n/k$-justifying groups and JR committees for general instances.
When the voters' preferences are drawn according to the standard \emph{impartial culture (IC)} model, in which each voter approves each candidate independently with probability $p$, we establish a sharp threshold on the group size: above this threshold, all groups are likely to be $n/k$-justifying, while below the threshold, no group is likely to be.
In particular, the threshold is below $k/2$ for every value of $p$, thereby providing a theoretical explanation of Bredereck et al.'s findings~\cite{BFKN19}.
Our result also implies that with high probability, the number of JR committees is very large, which means that the JR condition is not as stringent as it may seem.
On the other hand, we show that, in the worst case, there may be very few JR committees: their number can be as small as $m-k+1$ (where $m$ denotes the number of candidates), and this is tight.

Next, in \Cref{sec:instance-specific}, we focus on the problem of computing a small $n/k$-justifying group for a given instance.
While this problem is NP-hard to approximate to within a factor of $o(\ln n)$ even in the case $n = k$ (since it is equivalent to the well-known {\sc Set Cover} problem in that case\footnote{In \Cref{app:jr-inapprox}, we extend this hardness to the case $n/k > 1$.}), we show that the simple \emph{GreedyCC} algorithm \cite{LB11,nem-wol-fis:j:submodular} returns an $n/k$-justifying group whose size is at most $O(\sqrt{n})$ times the optimal size; moreover, this factor is asymptotically tight.
We then devise a new greedy algorithm, \emph{GreedyCandidate}, with approximation ratio $O(\log(mn))$. There are several applications of multiwinner voting where
the number of candidates $m$ is either smaller or not much larger than the number of voters $n$; for such applications, the approximation ratio of GreedyCandidate is much better than that of GreedyCC.
Further, we show that if the voters' preferences admit a \emph{tree representation}, an optimal solution can be found in polynomial time.
The tree representation condition is known to encompass several other preference restrictions \cite{yang2019tree}; interestingly, we show that it also generalizes a recently introduced class of restrictions called \emph{1D-VCR} \cite{GBSF21}.

While small $n/k$-justifying groups are interesting in their own right given that they offer a high degree of representation relative to their size, an important benefit of finding such a group is that one can complement it with other candidates to obtain a JR committee with properties that one desires---the smaller the group, the more freedom one has in choosing the remaining members of the committee.
We illustrate this with a common consideration in committee selection: gender balance.\footnote{Bredereck et al.~\cite{BFI18} studied maximizing objective functions of committees subject to gender balance and other diversity constraints, but did not consider JR.}
In \Cref{sec:gender-balance}, we show that although it is easy to find a JR committee with at least one member of each gender, computing or even approximating the smallest gender imbalance subject to JR is NP-hard.
Nevertheless, in \Cref{sec:experiment}, we demonstrate through experiments that both GreedyCC and GreedyCandidate usually find an $n/k$-justifying group of size less than $k/2$; by extending such a group, we obtain a gender-balanced JR committee in polynomial time.
In addition, we experimentally verify our result from \Cref{sec:general-guarantees} in the IC model, and perform analogous experiments in two Euclidean models.

\section{Preliminaries}

There is a finite set of candidates $C = \{c_1,\dots,c_m\}$  and 
a finite set of voters $N=[n]$, where we write $[t] := \{1,\dots,t\}$ for any positive integer $t$.
Each voter $i\in N$ submits a non-empty ballot
$A_i\subseteq C$,  and the goal is to select a committee, which is a subset of $C$ of size $k$.
Thus, an instance $I$ of our problem 
can be described by a set of candidates $C$, 
a list of ballots $\calA=(A_1, \dots, A_n)$, and a positive integer 
$k \le m$; we write $I=(C, \calA, k)$.

We are interested in representing the voters according to their ballots.
Given an instance $I = (C, \calA, k)$ with $\calA=(A_1, \dots, A_n)$,  
we say that a group of voters $N'\subseteq N$ is {\em cohesive} if
$\cap_{i\in N'} A_i\neq\emptyset$. Further, we say that a committee $W$
{\em represents} a group of voters $N'\subseteq N$
if $W\cap A_i\neq\emptyset$ for some $i\in N'$. 
If candidate $c_j\in W$ is approved by voter $i\in N$, we say that $c_j$ \emph{covers} $i$.
We are now ready to state
the justified representation axiom of Aziz et al.~\cite{ABC17}.

\begin{definition}[JR]
\label{def:jr}
Given an instance $I = (C, \calA, k)$ with $\calA=(A_1, \dots, A_n)$,  
we say that a committee $W\subseteq C$ of size $k$ provides 
{\em justified representation (JR)}
for $I$ if it represents 
every cohesive group of voters $N'\subseteq N$ 
such that $|N'|\ge n/k$. 
We refer to such a committee as a \emph{JR committee}.
\end{definition}

More generally, we can extend the JR condition to groups of fewer than $k$ candidates (the requirement that this group represents every cohesive group of at least $n/k$ voters is with respect to the \emph{original} parameter $k$).
Bredereck et al.~\cite{BFKN19} called such a group of candidates an \emph{$n/k$-justifying group}.

A simple yet important algorithm in this setting is \emph{GreedyCC} \cite{LB11,nem-wol-fis:j:submodular}.
We consider a slight modification of this algorithm.
Our algorithm starts with the empty committee and iteratively adds one candidate at a time.
At each step, if there is still an unrepresented cohesive group of size at least $n/k$, the algorithm identifies a largest such group and adds a common approved candidate of the group to the committee.
If no such group exists, the algorithm returns the current set of candidates, which is $n/k$-justifying by definition.
It is not hard to verify that (our version of) GreedyCC runs in polynomial time and outputs an $n/k$-justifying group of size at most $k$.
Sometimes we may let the algorithm continue by identifying a largest unrepresented cohesive group (of size smaller than $n/k$) and adding a common approved candidate of the group.

\section{General Guarantees}
\label{sec:general-guarantees}

In order to be $n/k$-justifying, a group may need to include $k$ candidates in the worst case: this happens, e.g., when $n$ is divisible by $k$, the first $k$ candidates are approved by disjoint sets of $n/k$ voters each, and the remaining $m-k$ candidates are only approved by one voter each.
However, many instances admit much smaller $n/k$-justifying groups.
Indeed, in the extreme case, if there is no cohesive group of voters, the empty group already suffices.
It is therefore interesting to ask what happens in the average case.
We focus on the well-studied \emph{impartial culture (IC)} model, in which each voter approves each candidate independently with probability $p$.
If $p=0$, the empty group is already $n/k$-justifying, while if $p=1$, any singleton group is sufficient.
For each $p$, we establish a sharp threshold on the group size: above this threshold, all groups are likely to be $n/k$-justifying, while below the threshold, it is unlikely that any group is $n/k$-justifying.

\begin{theorem} 
\label{thm:average-case}
Suppose that $m$ and $k$ are fixed, and let $p\in(0,1)$ be a real constant and $s\in [0,k]$ an integer constant.
Assume that the votes are distributed according to the IC model with parameter $p$.

\begin{enumerate}[(a)]
\item If $p(1-p)^s < 1/k$, then with high probability as $n\rightarrow\infty$, every group of $s$ candidates is $n/k$-justifying.
\item If $p(1-p)^s > 1/k$, then with high probability as $n\rightarrow\infty$, no group of $s$ candidates is $n/k$-justifying.
\end{enumerate}
\end{theorem}

Here, ``with high probability'' means that the probability converges to $1$ as $n\rightarrow\infty$.
To prove this result, we will make use of the following standard probabilistic bound.

\begin{lemma}[Chernoff bound] \label{lem:chernoff}
Let $X_1, \dots, X_t$ be independent random variables taking values in $[0, 1]$, and let $S := X_1 + \cdots + X_t$. Then, for any $\delta \in [0, 1]$,
$$\Pr[S \geq (1 + \delta)\E[S]] \leq \exp\left(\frac{-\delta^2 \E[S]}{3}\right)$$
and
$$\Pr[S \leq (1 - \delta)\E[S]] \leq \exp\left(\frac{-\delta^2 \E[S]}{2}\right).$$
\end{lemma}

\begin{proof}[of \Cref{thm:average-case}]
(a)
Let $p(1-p)^s = 1/k-\varepsilon$ for some constant $\varepsilon$, and consider any group $W\subseteq C$ of size $s$.
We claim that for any candidate $c\not\in W$, with high probability as $n\rightarrow\infty$, the number of voters who approve $c$ but do not
approve any of the candidates in $W$ is less than $n/k$.
Since $m$ is constant, once this claim is established, we can apply the union bound over all candidates outside $W$ to show that $W$ is likely to be $n/k$-justifying.
Then, we apply the union bound over all (constant number of) groups of size $s$.

Fix a candidate $c\not\in W$.
For each $i\in [n]$, let $X_i$ be an indicator random variable that indicates whether voter~$i$ approves $c$ and none of the candidates in $W$; $X_i$ takes the value $1$ if so, and $0$ otherwise.
Let $X := \sum_{i=1}^n X_i$.
We have $\E[X_i] = p(1-p)^s = 1/k-\varepsilon$ for each $i$, and so $\E[X] = n(1/k-\varepsilon)$.
By Lemma~\ref{lem:chernoff}, it follows that
\[
\Pr\left[X \ge \frac{n}{k}\right] \le \exp\left(-\frac{\delta^2 n\left(\frac{1}{k}-\varepsilon\right)}{3}\right),
\]
where $\delta := \min\{1, k\varepsilon/(1-k\varepsilon)\}$ is constant.
This probability converges to $0$ as $n\rightarrow\infty$, proving the claim. 

(b)
Let $p(1-p)^s = 1/k+\varepsilon$ for some constant $\varepsilon$.
First, suppose for contradiction that $s = k$.
The derivative of $f(p):=p(1-p)^k$ is $f'(p) = (1-p)^{k-1}(1-p(k+1))$, so $f(p)$ attains its maximum at $p^*=\frac{1}{k+1}$, where $f(p^*) = \frac{k^k}{(k+1)^{k+1}}< \frac{1}{k}$, a contradiction.
Hence $s < k$.

Consider any group $W\subseteq C$ of size $s$.
We claim that for any candidate $c\not\in W$ (such a candidate exists because $s < k$), with high probability as $n\rightarrow\infty$, the number of voters who approve $c$ but do not approve any of the candidates in $W$ is greater than $n/k$.
When this is the case, $W$ is not $n/k$-justifying.
We then apply the union bound over all possible groups $W$.

Fix a candidate $c\not\in W$, and define the random variables $X_1,\dots,X_n$ and $X$ as in part (a).
We have $\E[X_i] = p(1-p)^s = 1/k+\varepsilon$ for each $i$, and so $\E[X] = n(1/k+\varepsilon)$.
By Lemma~\ref{lem:chernoff}, it follows that
\[ \Pr\left[X \le \frac{n}{k}\right] \le \exp\left(-\frac{\delta^2 n\left(\frac{1}{k}+\varepsilon\right)}{2}\right),
\]
where $\delta := k\varepsilon/(1+k\varepsilon)$ is constant.
This probability converges to $0$ as $n\rightarrow\infty$, proving the claim.
\hfill $\square$ 
\end{proof}

Theorem~\ref{thm:average-case} implies that if $p < 1/k$, then the empty group is already $n/k$-justifying with high probability, because there is unlikely to be a sufficiently large cohesive group of voters.
On the other hand, when $p > 1/k$, the threshold for the required group size $s$ occurs when $p(1-p)^s = 1/k$, i.e., $s = -\log_{1-p}(kp)$.
For $k = 10$, the maximum $s$ occurs at $p\approx 0.24$, where we have $s\approx 3.19$.
This means that for every $p\in[0,1]$, an arbitrary group of size $4$ is likely to be $n/k$-justifying.
Interestingly, the threshold for $s$ never exceeds $k/2$ regardless of $p$.

\begin{proposition}
\label{prop:average-case-half}
Suppose that $m$ and $k$ are fixed, and let $p\in(0,1)$ be a real constant and $s\ge k/2$ an integer constant.
Assume that the votes are distributed according to the IC model with parameter $p$.
Then, with high probability as $n\rightarrow\infty$, every group of size $s$ is $n/k$-justifying.
\end{proposition}

\begin{proof}
By part (a) of \Cref{thm:average-case}, it suffices to show that $p(1-p)^s < 1/k$ for all $p\in(0,1)$ and integers $s \ge k/2$. 
As in the analysis of part (b) of \Cref{thm:average-case}, the function $f(p) := p(1-p)^s$ attains its maximum at $p^* = \frac{1}{s+1}$, where $f(p^*) = \frac{s^s}{(s+1)^{s+1}} = \frac{1}{s(1 + 1/s)^{s + 1}}$. 
By Bernoulli's inequality, we have $(1 + 1/s)^{s + 1} \geq 1 + (s + 1)/s > 2$. 
It follows that 
\[
p(1-p)^s \leq f(p^*) < \frac{1}{2s} \leq \frac{1}{k},
\]
as desired.
\hfill $\square$ 
\end{proof}

We remark that the proposition would not hold if we were to replace $k/2$ by $k/3$: indeed, for $k=15$, the maximum $s$ occurs at $p\approx 0.17$, where we have $s\approx 5.03 > 15/3$.

An implication of \Cref{prop:average-case-half} is that under the IC model, with high probability, every size-$k$ committee provides JR.
This raises the question of whether the number of JR committees is large even in the worst case.
The following example shows that the answer is negative: when $n$ is divisible by $k$, the number of JR committees can be as small as $m-k+1$.

\begin{example}
\label{ex:worst-case}
Assume that $n$ is divisible by $k$.
Consider an instance $I = (C, \calA, k)$ where 
\begin{itemize}
\item[] $A_1=\dots=A_{\frac{n}{k}}=\{c_1\}$;
\item[] $A_{\frac{n}{k}+1}=\dots=A_{\frac{2n}{k}}=\{c_2\}$;
\item[] $\vdots$
\item[] $A_{\frac{(k-2)n}{k}+1}=\dots=A_{\frac{(k-1)n}{k}}=\{c_{k-1}\}$;
\item[] $A_{\frac{(k-1)n}{k}+1}=\dots=A_{n}=\{c_k,c_{k+1},\dots,c_m\}$.
\end{itemize}
A JR committee must include $c_1,\dots,c_{k-1}$; for the last slot, any of the remaining $m-k+1$ candidates can be chosen.
Hence, there are exactly $m-k+1$ JR committees.
\end{example}

We complement \Cref{ex:worst-case} by establishing that, as long as every candidate is approved by at least one voter, there are always at least $m-k+1$ JR committees.\footnote{The condition that every candidate is approved by at least one voter is necessary. Indeed, if the last approval set in \Cref{ex:worst-case} is changed from $\{c_k,c_{k+1},\dots,c_m\}$ to $\{c_k\}$, then there is only one JR committee: $\{c_1,c_2,\dots,c_k\}$.}
This matches the upper bound in Example~\ref{ex:worst-case} and improves upon the bound of $m/k$ by Bredereck et al.~\cite[Thm.~3]{BFKN19}.
Moreover, the bound holds regardless of whether $n$ is divisible by~$k$.

\begin{theorem}
\label{thm:worst-case}
For every instance $I=(C, \calA, k)$ such that every candidate in $C$ is approved by some voter, at least $m-k+1$ committees of size $k$ provide JR.
\end{theorem}

\begin{proof}
We run GreedyCC for $k-1$ steps. 
If the resulting group (of size $k-1$) is already $n/k$-justifying, we can choose any of the remaining $m-k+1$ candidates as the final member of the committee.
Hence, assume that the group after $k-1$ steps is not $n/k$-justifying.
This means that each of the first $k-1$ candidates covers exactly $n/k$ voters (these sets of voters are disjoint), and the remaining $n/k$ voters are covered by another candidate. 
In particular, $n$ is divisible by $k$.
Call these blocks of $n/k$ voters $B_1,\dots,B_k$, and assume without loss of generality that the corresponding candidates are $c_1,\dots,c_k$, respectively.
For each of the remaining $m-k$ candidates, the candidate is approved by at least one voter, say in block $B_i$, so we can combine the candidate with $\{c_1,\dots,c_{i-1},c_{i+1},\dots,c_k\}$ to form a JR committee.
This yields $m-k$ distinct JR committees.
Finally, the committee $\{c_1,\dots,c_k\}$ also provides JR and differs from all of the above commitees.
It follows that there are at least $m-k+1$ JR committees.
\hfill $\square$ 
\end{proof}

\section{Instance-Specific Optimization}
\label{sec:instance-specific}

As we have seen in \Cref{sec:general-guarantees}, several instances admit an $n/k$-justifying group of size much smaller than the worst-case size~$k$.
However, the problem of computing a minimum-size $n/k$-justifying group is NP-hard to approximate to within a factor of $o(\ln n)$ even when $n = k$ (see \Cref{sec:contribution}).
In this section, we address the question of how well we can approximate such a group in polynomial time.

\subsection{GreedyCC}
\label{sec:GreedyCC}

A natural approach to computing a small $n/k$-justifying group is to simply run (our variant of) GreedyCC, stopping as soon as the current group is $n/k$-justifying.
However, as the following example shows, the output of this algorithm may be $\Theta(\sqrt{n})$ times larger than the optimal solution.

\begin{example}
\label{ex:greedy-is-bad}
Let $n=k^2$ and $m=2k$, for some $k \ge 3$.
Consider an instance $I = (C, \calA, k)$ where
\begin{itemize}
\item[] $A_1 = \dots = A_{k-1} = \{c_1\}$;
\item[] $A_{(k-1)+1} = \dots = A_{2(k-1)} = \{c_2\}$;
\item[] $\vdots$
\item[] $A_{(k-2)(k-1)+1} = \dots = A_{(k-1)^2} = \{c_{k-1}\}$;
\item[] $A_{(k-1)^2+1} = \{c_1,c_k\}$;
\item[] $A_{(k-1)^2+2} = \{c_2,c_k\}$;
\item[] $\vdots$
\item[] $A_{k(k-1)} = \{c_{k-1},c_k\}$;
\item[] $A_{k(k-1)+1} = \{c_{k+1}\}$;
\item[] $A_{k(k-1)+2} = \{c_{k+2}\}$;
\item[] $\vdots$
\item[] $A_{k^2} = \{c_{2k}\}$.
\end{itemize}
Since $c_1,\dots,c_{k-1}$ are each approved by pairwise disjoint groups of $k$ voters, while $c_k$ is approved by $k-1$ voters, GreedyCC outputs the group $\{c_1,\dots,c_{k-1}\}$.
However, the singleton group $\{c_k\}$ is already $n/k$-justifying.
The ratio between the sizes of the two groups is $(k-1) \in \Theta(\sqrt{n})$.
\end{example}

It turns out that \Cref{ex:greedy-is-bad} is already a worst-case scenario for GreedyCC, up to a constant factor.

\begin{theorem}
\label{thm:greedy-approx}
For every instance $I=(C, \calA, k)$, GreedyCC outputs an $n/k$-justifying group at most $\sqrt{2n}$ times larger than a smallest $n/k$-justifying group.
\end{theorem}

\begin{proof}
Assume without loss of generality that $\copt := \{c_1,c_2,\dots,c_t\}$ is a smallest $n/k$-justifying group; our goal is to show that GreedyCC selects at most $\sqrt{2n}\cdot t$ candidates.
For $j\in[t]$, we say that a candidate $c_r\in C$ (possibly $r=j$) chosen by GreedyCC \emph{crosses} $c_j$ if $c_r$ is approved by some voter $i$ who approves $c_j$ and does not approve any candidate chosen by GreedyCC up to the point when $c_r$ is selected.
Note that each candidate selected by GreedyCC must cross some candidate in $\copt$---indeed, if not, the cohesive group of at least $n/k$ voters that forces GreedyCC to select the candidate would not be represented by $\copt$, contradicting the assumption that $\copt$ is $n/k$-justifying.

Now, we claim that for each $j\in [t]$, the candidate $c_j$ can be crossed by at most $\sqrt{2n}$ candidates in the GreedyCC solution; this suffices for the desired conclusion.
The claim is immediate if $c_j$ is approved by at most $\sqrt{2n}$ voters, because each candidate selected by GreedyCC that crosses $c_j$ must cover a new voter who approves $c_j$.
Assume therefore that $c_j$ is approved by more than $\sqrt{2n}$ voters, and suppose for contradiction that $c_j$ is crossed by more than $\sqrt{2n}$ candidates in the GreedyCC solution. Denote these candidates by $c_{\ell_1},\dots,c_{\ell_s}$ in the order that GreedyCC selects them, where $s > \sqrt{2n}$.
Notice that for each $i\in[s-1]$, when GreedyCC selects $c_{\ell_i}$, it favors $c_{\ell_i}$ over $c_j$, which would cover at least $s-i+1$ uncovered voters (i.e., the ``crossing points'' of $c_{\ell_{i}},\dots,c_{\ell_s}$ with $c_j$).
Hence, $c_{\ell_i}$ itself must cover at least $s-i+1$ uncovered voters.
Moreover, $c_{\ell_s}$ covers at least one uncovered voter.
This means that the total number of voters is at least $s+(s-1)+\dots+1 = s(s+1)/2 > n$, a contradiction.
\hfill $\square$ 
\end{proof}

\subsection{GreedyCandidate}
\label{sec:GreedyCandidate}

Next, we present a different greedy algorithm, which provides an approximation ratio of $\ln(mn) + 1$.
Note that this ratio is asymptotically better than the ratio of GreedyCC in the range $m \in 2^{o(\sqrt{n})}$; several practical elections fall under this range, since the number of candidates is typically smaller or, at worst, not much larger than the number of voters (e.g., when Internet users vote upon world heritage site candidates or students elect student council members).

To understand our new algorithm, recall that GreedyCC can be viewed as a greedy covering algorithm, where the goal is to pick candidates to cover the voters. 
Our new algorithm instead views the problem as ``covering'' the candidates.
Specifically, for a set of candidates $W\subseteq C$ to be an $n/k$-justifying group, all but at most $\ell := \lceil n/k \rceil - 1$ of the voters who approve each candidate in $C$ must be ``covered'' by $W$. 
In other words, each candidate $c\in C$ must be ``covered'' at least $[|B_{c}^0| - \ell]_+$ times, where $B_c^0$ denotes the set of voters who approve $c$ and we use the notation $[x]_+$ as a shorthand for $\max\{x, 0\}$. 
Our algorithm greedily picks in each step a candidate whose selection would minimize the corresponding potential function, $\sum_{c' \in C} [|B_{c'}| - \ell]_+$, where $B_{c'}$ denotes the set of voters who approve $c'$ but do not approve any candidate selected by the algorithm thus far.
The pseudocode of the algorithm, which we call GreedyCandidate, is presented as Algorithm~\ref{alg:greedy-cand}.
One can check that GreedyCandidate runs in polynomial time.

\begin{algorithm}[t]
  \DontPrintSemicolon
  \SetAlgoLined
  \KwIn{An instance $(C,\calA, k)$, where $A = \{A_1, \ldots, A_n\}$}
  \KwOut{An $n/k$-justifying group $W \subseteq C$}
  \BlankLine
  $\ell \leftarrow \lceil n/k \rceil - 1$\;
  $W \leftarrow \emptyset$\;
  \For{$c \in C$}{
    $B_c \leftarrow \{i \in [n] : c \in A_i\}$\;
  }
  \While{there exists $c \in C$ such that $|B_c| > \ell$}{
    \For{$c \in C$}{
        $u_c \leftarrow \sum_{c' \in C} ([|B_{c'}| - \ell]_+ - [|B_{c'} \setminus B_c| - \ell]_+)$\;
    }
    $c^* \leftarrow \arg\max_{c \in C} u_c$\;
    $W \leftarrow W \cup \{c^*\}$\;
    \For{$c \in C$}{
        $B_c \leftarrow (B_c \setminus B_{c^*})$
    }
    }
    \Return $W$
    \caption{\label{alg:greedy-cand}GreedyCandidate}
\end{algorithm}

\begin{theorem}
\label{thm:greedy-candidate-approx}
For every instance $I=(C, \calA, k)$, GreedyCandidate outputs an $n/k$-justifying group that is at most $(\ln(mn) + 1)$ times larger than a smallest $n/k$-justifying group.
\end{theorem}

\begin{proof}
First, note that whenever $|B_c| \leq \ell$ for all $c \in C$, every unrepresented cohesive group has size at most $\ell < n/k$, meaning that the output $W$ of our algorithm is indeed an $n/k$-justifying group.

Next, let us bound the size of the output $W$. 
Assume without loss of generality that $\copt := \{c_1,c_2,\dots,c_t\}$ is a smallest $n/k$-justifying group. 
If $t = 0$, then the while-loop immediately terminates and the algorithm outputs $W = \emptyset$. 
Thus, we may henceforth assume that $t \ge 1$.

For each $c\in C$, denote by $B^i_c$ the set $B_c$ after the $i$-th iteration of the while-loop (so $B^0_c$ is simply the set of voters who approve $c$).
Let $\psi_i$ denote the potential $\sum_{c' \in C} [|B_{c'}| - \ell]_+$ after the $i$-th iteration (so $\psi_0$ is the potential at the beginning). 
We will show that this potential decreases by at least a factor of $1 - 1/t$ with each iteration; more formally,
\begin{align} \label{eq:potential-reduction}
\psi_i \leq \left(1 - \frac{1}{t}\right) \cdot \psi_{i - 1}
\end{align}
for each $i\ge 1$.

Before we prove~\eqref{eq:potential-reduction}, let us show how we can use it to bound $|W|$. 
To this end, observe that when the potential is less than $1$, the while-loop terminates. 
This means that $\psi_{|W| - 1} \geq 1$. 
Furthermore, we have $\psi_0 \leq \sum_{c' \in C} |B^0_{c'}| \leq mn$. 
Applying~\eqref{eq:potential-reduction}, we get
\begin{align*}
1 \leq \psi_{|W| - 1} \leq \cdots &\leq \left(1 - \frac{1}{t}\right)^{|W| - 1} \cdot \psi_0 
\leq e^{-\frac{|W| - 1}{t}} \cdot mn,
\end{align*}
where for the last inequality we use the bound $1+x\le e^x$, which holds for any $x\in\mathbb{R}$.
Rearranging, we arrive at $|W| \leq 1 + t \ln(mn) \leq t(\ln(mn) + 1)$, as desired.

We now return to proving~\eqref{eq:potential-reduction}. 
Our assumption that $\copt$ is an $n/k$-justifying group implies that $\left|B^0_{c'} \setminus \left(\bigcup_{j=1}^t B^0_{c_j}\right)\right| \leq \ell$ for all $c' \in C$.
In each iteration, the algorithm replaces $B_c$ by $B_c \setminus B_{c^*}$ for all $c$, so we also have $\left|B^i_{c'} \setminus \left(\bigcup_{j=1}^t B^i_{c_j}\right)\right| \leq \ell$ for all $c' \in C$ and $i$.
Fix any $i\ge 1$, let $q := i-1$, and let $c^*$ be the candidate chosen in the $i$-th iteration.
From the definition of $c^*$, we have
\begin{align}
\psi_{i - 1} - \psi_{i} 
&= \sum_{c' \in C} \left(\left[|B^q_{c'}| - \ell\right]_+ - \left[|B^q_{c'} \setminus B^q_{c^*}| - \ell\right]_+\right) \nonumber \\
&\geq \frac{1}{t} \sum_{j=1}^t \sum_{c' \in C} \left(\left[|B^q_{c'}| - \ell\right]_+ - \left[|B^q_{c'} \setminus B^q_{c_j}| - \ell\right]_+\right) \nonumber \\
&= \psi_{i-1} - \sum_{c' \in C} \left(\frac{1}{t} \sum_{j=1}^t \left[|B^q_{c'} \setminus B^q_{c_j}| - \ell\right]_+\right). \label{eq:expanded-diff}
\end{align}
Consider any $c' \in C$. We claim that 
\begin{align} \label{eq:single-term-potential-loss}
\sum_{j=1}^t \left[|B^q_{c'} \setminus B^q_{c_j}| - \ell\right]_+ \leq (t - 1) \cdot \left[|B^q_{c'}| - \ell\right]_+.
\end{align}
To see that~\eqref{eq:single-term-potential-loss} holds, consider the following two cases:
\begin{itemize}
\item \underline{Case 1}: $|B^q_{c'} \setminus B^q_{c_{j'}}| \leq \ell$ for some $c_{j'} \in C_{\opt}$. 
We may assume without loss of generality that $j' = t$.
We have
\begin{align*}
\sum_{j=1}^t \left[|B^q_{c'} \setminus B^q_{c_j}| - \ell\right]_+ &= \sum_{j=1}^{t - 1} \left[|B^q_{c'} \setminus B^q_{c_j}| - \ell\right]_+ \\
&\leq \sum_{j=1}^{t - 1} \left[|B^q_{c'}| - \ell\right]_+ 
= (t - 1) \cdot \left[|B^q_{c'}| - \ell\right]_+.
\end{align*}
\item \underline{Case 2}: $|B^q_{c'} \setminus B^q_{c_j}| > \ell$ for all $c_j \in C_{\opt}$. 
This means that $|B^q_{c'}| > \ell$, and
\begin{align*}
\sum_{j=1}^t \left[|B^q_{c'} \setminus B^q_{c_j}| - \ell\right]_+ &= \sum_{j=1}^t (|B^q_{c'} \setminus B^q_{c_j}| - \ell) \\
&= t(|B^q_{c'}| - \ell) - \sum_{j=1}^t |B^q_{c'} \cap B^q_{c_j}| \\
&\leq t(|B^q_{c'}| - \ell) - (|B^q_{c'}| - \ell) = (t - 1) \cdot \left[|B^q_{c'}| - \ell\right]_+,
\end{align*}
where we use $\left|B^q_{c'} \setminus \left(\bigcup_{j=1}^t B^q_{c_j}\right)\right| \leq \ell$ for the inequality.
\end{itemize}
Hence, in both cases, \eqref{eq:single-term-potential-loss} holds. 
Plugging this back into~\eqref{eq:expanded-diff}, we get
\begin{align*}
\psi_{i - 1} - \psi_{i}
\geq \psi_{i - 1} - \frac{t - 1}{t} \cdot \sum_{c' \in C} \left[|B^q_{c'}| - \ell\right]_+ 
= \frac{1}{t} \cdot \psi_{i - 1}.
\end{align*}
This implies~\eqref{eq:potential-reduction} and completes our proof. 
\hfill $\square$ 
\end{proof}

Although we do not know whether an efficient $O(\ln(n))$-approximation algorithm exists, we show next that by combining~\Cref{thm:greedy-candidate-approx} with a brute-force approach, we can arrive at a quasi-polynomial-time\footnote{Recall that a running time is said to be \emph{quasi-polynomial} if it is of the form $\exp(\log^{O(1)} I)$, where $I$ denotes the input size (in our case, $I = (nm)^{O(1)}$).} algorithm that has an approximation ratio of $O(n^{\delta})$ for any constant $\delta > 0$.

\begin{theorem}
\label{thm:quasi-polynomial}
For any constant $\delta \in (0, 1)$ there exists an $\exp(\log^{O(1)} (nm))$-time algorithm that, on input $I=(C, \calA, k)$, outputs an $n/k$-justifying group that is at most $O(n^{\delta})$ times larger than a smallest $n/k$-justifying group.
\end{theorem}

\begin{proof}
If $m \leq 2^{n^\delta}$, then we may simply run GreedyCandidate, which runs in polynomial time and yields an approximation ratio of $\ln(nm) + 1 \in O(n^{\delta})$. 

Otherwise, we iterate over all subsets $N' \subseteq N$.
For each $N'$, we run GreedyCC on $N'$ until all voters in $N'$ are covered (i.e., we stop only when no voter in $N'$ remains unrepresented); denote the resulting set of candidates by $W_{N'}$. 
Finally, among the $2^{|N|}$ sets, we output a smallest one that is $n/k$-justifying with respect to $N$. 
Notice that the running time of our algorithm is $O(2^n \cdot (nm)^{O(1)}) \in \exp(\log^{O(1/\delta)} m)$; this follows from $m > 2^{n^{\delta}}$. 
To analyze the approximation guarantee, let $\copt$ denote a smallest $n/k$-justifying group, and $N^* := \bigcup_{c \in \copt} B_c$ be the set of voters covered by $\copt$, where $B_c$ denotes the set of voters who approve $c$. 
When $N' = N^*$, we have that $W_{N'}$ is an $n/k$-justifying group and, from standard analyses of the greedy set cover algorithm,\footnote{See, for example, Chapter~2 in the book by Vazirani~\cite{V03}.} we have $|W_{N'}| \leq (\ln n + 1) \cdot |\copt|$. 
In other words, the approximation ratio is at most $\ln n + 1 \in O(n^\delta)$, as claimed.
\hfill $\square$ 
\end{proof}

\subsection{Tree Representation}
\label{sec:tree-representation}

Even though computing a smallest $n/k$-justifying group is NP-hard even to approximate, we show in this section that this problem becomes tractable if the instance admits a \emph{tree representation}.
An instance is said to admit a tree representation if its candidates can be arranged on a tree $T$ in such a way that the approved candidates of each voter form a subtree of $T$ (i.e., the subgraph of $T$ induced by each approval set is connected).
While the tree representation condition is somewhat restrictive, we remark that it is general enough to capture a number of other preference restrictions \cite[Fig.~4]{yang2019tree}.
In particular, we show in \Cref{app:restriction} that it encompasses a recently introduced class called \emph{1-dimensional voter/candidate range model (1D-VCR)}~\cite{GBSF21}.\footnote{Together with the results of Yang~\cite{yang2019tree} and Godziszewski et al.~\cite{GBSF21}, this
means that the tree representation also captures the \emph{candidate interval (CI)} and \emph{voter interval (VI)} domains, the two most commonly studied restrictions for the approval setting, introduced by Elkind and Lackner~\cite{EL15}.}

\begin{theorem}
\label{thm:tree-algo}
For every instance $I=(C,\mathcal{A},k)$ admitting a tree representation, a smallest $n/k$-justifying group can be computed in polynomial time. 
\end{theorem}

\begin{proof}
\begin{figure}[!ht]
\centering
\begin{tikzpicture}[scale=0.7]
\draw[line cap=round, green!20, line width=4mm] (4,4.5) -- (4,1.5);
\draw[line cap=round, green!20, line width=4mm] (4,4.5) -- (1,3) -- (0,1.5) -- (-0.5,0);
\draw[line cap=round, green!20, line width=4mm] (0,1.5) -- (0.5,0);
\draw[line cap=round, green!20, line width=4mm] (1,3) -- (2,1.5);
\draw[line cap=round, red!20, line width=4mm] (7,3) -- (6,1.5) -- (5.5,0);
\draw[line cap=round, red!20, line width=4mm] (7,3) -- (8,1.5) -- (7.5,0);
\draw[line cap=round, red!20, line width=4mm] (6,1.5) -- (6.5,0);
\draw[line cap=round, red!20, line width=4mm] (8,1.5) -- (8.5,0);
\draw (4,4.5) -- (4,3);
\draw (4,4.5) -- (1,3);
\draw (4,4.5) -- (7,3);
\draw (1,3) -- (0,1.5);
\draw (1,3) -- (2,1.5);
\draw (4,3) -- (4,1.5);
\draw (7,3) -- (6,1.5);
\draw (7,3) -- (8,1.5);
\draw (0,1.5) -- (-0.5,0);
\draw (0,1.5) -- (0.5,0);
\draw (6,1.5) -- (5.5,0);
\draw (6,1.5) -- (6.5,0);
\draw (8,1.5) -- (7.5,0);
\draw (8,1.5) -- (8.5,0);
\draw [fill] (4,4.5) circle [radius = 0.1];
\draw [fill] (4,3) circle [radius = 0.1];
\draw [fill] (1,3) circle [radius = 0.1];
\draw [fill] (7,3) circle [radius = 0.1];
\draw [fill] (0,1.5) circle [radius = 0.1];
\draw [fill] (2,1.5) circle [radius = 0.1];
\draw [fill] (4,1.5) circle [radius = 0.1];
\draw [fill] (6,1.5) circle [radius = 0.1];
\draw [fill] (8,1.5) circle [radius = 0.1];
\draw [fill] (-0.5,0) circle [radius = 0.1];
\draw [fill] (0.5,0) circle [radius = 0.1];
\draw [fill] (5.5,0) circle [radius = 0.1];
\draw [fill] (6.5,0) circle [radius = 0.1];
\draw [fill] (7.5,0) circle [radius = 0.1];
\draw [fill] (8.5,0) circle [radius = 0.1];
\node at (5.9,2.5) {\small $\overline{C_t}$};
\node at (2.1,4.2) {\small $T_t$};
\node at (7.2,3.2) {\small $v$};
\node at (8.8,0) {\small $w$};
\node at (8.3,1.5) {\small $z$};
\end{tikzpicture}
\caption{Illustration for the proof of Theorem~\ref{thm:tree-algo}.}
\label{fig:tree-algo}
\end{figure}
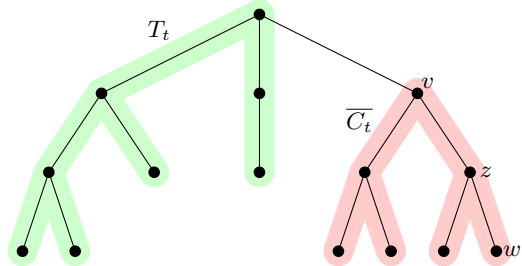

Let $T$ be a tree representation of $I$, i.e., for every $i\in N$ the set $A_i$ induces a subtree of $T$. 
Root $T$ at an arbitrary node, and define the depth of a node in~$T$ as its distance from the root node (so the root node itself has depth $0$). 
For each subtree $\widehat{T}$ of $T$, denote by $V(\widehat{T})$ the set of its nodes, and for each node $v\in V(\widehat{T})$, denote by $\widehat{T}^v$ the subtree of $\widehat{T}$ rooted at $v$ (i.e., $\widehat{T}^v$ contains all nodes in $\widehat{T}$ whose path towards the root of $\widehat{T}$ passes $v$).
The algorithm sets $W = \emptyset$ and proceeds as follows: 
\begin{enumerate}
    \item Select a node $v$ of maximum depth such that there exists a set $S$ of $\lceil n/k \rceil$ voters with the following two properties:
    \begin{enumerate}
    \item $A_i \subseteq V(T^v)$ for all $i \in S$;
    \item $\bigcap_{i \in S} A_i \neq \emptyset$.
    \end{enumerate}
    If no such node exists, delete all candidates from $C$, delete the remaining tree $T$, and return $W$.  \label{alg:trSmallestJRCommittee}
    \item Add $v$ to $W$, remove all voters $i$ such that $A_i \cap V(T^v) \neq \emptyset$ from $\mathcal{A}$, and delete $V(T^v)$ from $C$ and $T^v$ from $T$. Go back to Step~\ref{alg:trSmallestJRCommittee}. 
\end{enumerate}

Except for the last round, the algorithm adds one candidate to the set $W$ in each round, so it runs for $|W|+1$ rounds, where we slightly abuse notation and use $W$ to refer to the final output from now on. 
Each round can be implemented in polynomial time---indeed, for each node $v$, we can consider the sets $A_i$ that are contained in $V(T^v)$ and check whether some node in $V(T^v)$ appears in at least $\lceil n/k\rceil$ of these sets.

We now establish the correctness of the algorithm.
For each round $t \in \{0,1,\dots,|W|+1\}$, we define $W_t$ to be the set of candidates selected by the algorithm up to and including round~$t$, and $T_t$ to be the remaining tree after round~$t$, where round~$0$ refers to the point before the execution of the algorithm (so $W_0 = \emptyset$ and $T_0 = T$).
We also define $\overline{C_t} := V(T) \setminus V(T_t)$ to be the set of candidates deleted up to and including round~$t$.
See Figure~\ref{fig:tree-algo} for an illustration.

\begin{claim}
After each round $t \in \{0,1,\dots,|W|+1\}$,

(i) there exists a smallest $n/k$-justifying group $W'$ of the original instance~$I$ such that $\overline{C_t} \cap W' = W_t$, and

(ii) for each $i\in N$, at least one of the following three relations holds: $A_i \subseteq \overline{C_t}$, $A_i \cap \overline{C_t} = \emptyset$, or $A_i \cap W_t \neq \emptyset$. 
\end{claim}

\begin{proof}[of Claim]
We prove the Claim by induction on $t$.
For $t=0$, we have $\overline{C_0}=W_0=\emptyset$, so $A_i\cap\overline{C_0} = \emptyset$ for all $i\in N$, and both (i) and (ii) hold trivially. 
Now consider any $t \in [|W|+1]$ and assume that the Claim holds for $t-1$. 

\underline{Case $1$}: $t < |W|+1$. Let $v$ be the candidate selected in this round and $S$ be the corresponding set of voters in the algorithm.
Then, $\overline{C_t} = \overline{C_{t-1}} \cup V(T^v_{t-1})$ and $W_t = W_{t-1} \cup \{v\}$. 
Let $W'$ be a smallest $n/k$-justifying group with $\overline{C_{t-1}} \cap W' = W_{t-1}$ as guaranteed by the induction hypothesis. 
If $V(T^v_{t-1}) \cap W' = \{v\}$, then statement (i) of the Claim follows by choosing the same set $W'$. 
Assume therefore that $V(T^v_{t-1}) \cap W' \neq \{v\}$.
If $V(T^v_{t-1}) \cap W' = \emptyset$, then $W'$ does not represent the cohesive group of voters $S$ of size $\lceil n/k\rceil$, a contradiction.
Hence, there exists a candidate $w \in (V(T_{t-1}^v) \cap W')\setminus \{v\}$; let $z$ be the parent of $w$ (possibly $z=v$). 
See Figure~\ref{fig:tree-algo}.

We will show that $(W' \setminus \{w\}) \cup \{z\}$ is a smallest $n/k$-justifying group. 
Since its size is at most the size of $W'$, if it is an $n/k$-justifying group, then it is also a smallest one.
Assume for contradiction that the group is not $n/k$-justifying. 
This means there exists a group of voters $S'$ of size $\lceil n/k\rceil$ such that: (1) the voters in $S'$ approve a common candidate $y$; (2) at least one voter in $S'$ approves $w$; and (3) none of the voters in $S'$ approves any candidate in $W_{t-1} \cup \{z\}$. 

Observe that for a voter $j\in S'$ who approves $w$, we know that $A_j \cap \overline{C_{t-1}} = \emptyset$ by statement (ii) of the Claim for $t-1$.
Moreover, $A_j \subseteq V(T_{t-1}^w) \cup \overline{C_{t-1}}$ since $j$ does not approve~$z$. 
Combining the previous two sentences, we have $A_j \subseteq V(T_{t-1}^w)$. 
As a result, $y \in V(T_{t-1}^w)$ as well, and by the same arguments as for $j$ we get that $A_i \subseteq V(T_{t-1}^w)$ for all $i \in S'$. 
However, this is a contradiction to the choice of $v$ in the algorithm. 

Applying this argument (as we did between $w$ and $z$) repeatedly until we reach $v$, we obtain that there exists a smallest $n/k$-justifying group $W'$ with $\overline{C_t} \cap W' = 
W_t$.
This proves statement (i) of the Claim.

For (ii), we argue that for each voter $i \in N$, at least one of the relations $A_i \subseteq \overline{C_t}$, $A_i \cap \overline{C_t} = \emptyset$, and $A_i \cap W_t \neq \emptyset$ holds. 
Consider a voter $i$ for whom neither $A_i \subseteq \overline{C_t}$ nor $A_i \cap \overline{C_t} = \emptyset$ holds. 
If $A_i \cap \overline{C_{t-1}} \neq \emptyset$, then $A_i \cap W_{t-1} \neq \emptyset$ (and therefore $A_i \cap W_{t} \neq \emptyset$) follows from the induction hypothesis---indeed, since $A_i \not\subseteq \overline{C_t}$, it must be that $A_i \not\subseteq \overline{C_{t-1}}$.
Hence, we can assume that $A_i \cap \overline{C_{t-1}} = \emptyset$ and $A_i \cap \overline{C_t} \neq \emptyset$, which means that $A_i \cap V(T_{t-1}^v) \neq \emptyset$. 
Since $A_i$ is not a subset of $\overline{C_t}$, this implies that $v \in A_i$, and therefore $A_i \cap W_t \neq \emptyset$. 
This establishes (ii).

\underline{Case $2$}: $t=|W|+1$.
We will show that $W_{t-1}$ is an $n/k$-justifying group; assume for contradiction that it is not.
By the induction hypothesis, there exists an $n/k$-justifying group $W'$ such that $\overline{C_{t-1}} \cap W' = W_{t-1}$. 
In particular, there exists a group of voters $S'$ of size $\lceil n/k\rceil$ such that all of them approve a common candidate, at least one of them approves a candidate in $W'\setminus W_{t-1}$, and none of them approves any candidate in $W_{t-1}$.
Similarly to Case~1, we can argue that for all $i\in S'$, it holds that $A_i \subseteq C \setminus \overline{C_{t-1}}$. 
It follows that the root node of $T$ satisfies both conditions (a) and (b) in Step~1 of the algorithm, contradicting the fact that the algorithm terminated.
Hence, $W_{t-1}$ is an $n/k$-justifying group, and we can take $W' = W_{t-1}$ for statement (i) of the Claim.

For (ii), it suffices to observe that $\overline{C_t} = C$, which means that $A_i\subseteq\overline{C_t}$ for all $i\in N$.
\hfill $\square$ 
\end{proof}

As statement (i) of the claim holds in particular for $t=|W|+1$, in which case $\overline{C_t} = C$, this concludes the proof of the theorem. 
\hfill $\square$ 
\end{proof}

\section{Gender Balance}
\label{sec:gender-balance}

In the next two sections, we demonstrate that small $n/k$-justifying groups can be useful for obtaining JR committees with additional properties.
For concreteness, we consider a common desideratum: gender balance.
Indeed, in many candidate selection scenarios, it is preferable to have a balance with respect to the gender of the committee members.
Formally, assume that each candidate in $C$ belongs to one
of two types, male and female.  
For each committee
$W\subseteq C$, we define the \emph{gender imbalance} of $W$ as the
absolute value of the difference between the number of male
candidates and the number of female candidates in $W$.  A committee
is said to be \emph{gender-balanced} if its gender imbalance is $0$.

The following example shows that gender balance can be at odds with justified representation.

\begin{example}
\label{ex:balance}
Suppose that $n=k$ is even, each voter $i\in[n-1]$ only approves a male candidate $a_i$, while voter~$n$ approves female candidates $b_1,\dots,b_{n-1}$.
Any JR committee must contain all of $a_1,\dots,a_{n-1}$, and can therefore contain at most one female candidate.
But there exists a (non-JR) gender-balanced committee $\{a_1,\dots,a_{n/2},b_1,\dots,b_{n/2}\}$.
\end{example}

Example~\ref{ex:balance} is as bad as it gets: under very mild conditions, there always exists a JR committee with at least one representative of each gender.

\begin{theorem}
\label{thm:balance-one}
For every instance $I=(C, \calA, k)$ such that for each gender, some candidate of that gender is approved by at least one voter, there exists a JR committee with at least one member of each gender.
Moreover, such a committee can be computed in polynomial time.
\end{theorem}

\begin{proof}
As in the proof of \Cref{thm:worst-case}, we run GreedyCC for $k-1$ steps, and consider two cases.
If the resulting group of size $k-1$ is already $n/k$-justifying, we can choose the last member to be of the missing gender if necessary.
Else, we continue with the $k$-th step, and assume that the obtained committee is, say, all-female.
In this case, like in the proof of \Cref{thm:worst-case}, $n$ must be divisible by $k$, and there exist disjoint blocks of $n/k$ voters $B_1,\dots,B_k$ and candidates $c_1,\dots,c_k$ such that for each $j\in [k]$, all voters in block $B_j$ approve candidate $c_j$.
Take a male candidate who is approved by some voter, say, in block~$B_i$.
Then, the committee consisting of this candidate together with $c_1,\dots,c_{i-1},c_{i+1},\dots,c_k$ provides JR and contains members of both genders.
Since the construction of the blocks $B_1,\dots,B_k$ can be done in polynomial time, computing the committee also takes polynomial time.
\hfill $\square$ 
\end{proof}

In light of \Cref{thm:balance-one}, it is natural to ask for a JR committee with the lowest gender imbalance.
Unfortunately, our next result shows that deciding whether there exists a gender-balanced committee that provides JR, or even obtaining a close approximation thereof, is computationally hard.

\begin{theorem}
\label{thm:balance-hardness}
Even when $n=k$, there exists a constant $\varepsilon > 0$ such that distinguishing between the following two cases is NP-hard:
\begin{itemize}
\item (YES) There exists a gender-balanced JR committee;
\item (NO) Every JR committee has gender imbalance $\ge \varepsilon k$.
\end{itemize}
\end{theorem}

It follows from \Cref{thm:balance-hardness} that one cannot hope to obtain any finite (multiplicative) approximation of the gender imbalance.
To establish this hardness, we reduce from a special case of the {\sc Set Cover} problem.
Recall that in {\sc Set Cover}, we are given a universe $[u] = \{1,\dots,u\}$ and a collection $\calS = \{S_1, \dots, S_M\}$ of subsets of $[u]$. 
The goal is to select as few subsets as possible that together cover the universe; we use $\opt_{\setcov}(u, \calS)$ to denote the optimum of a {\sc Set Cover} instance $(u, \calS)$. 
We consider a special case of {\sc Set Cover} where $|S_1| = \dots = |S_M| = 3$; this problem is sometimes referred to as {\sc Exact Cover by 3-Sets (X3C)}.
We will need the following known APX-hardness of {\sc X3C}.\footnote{This hardness follows from the standard NP-hardness reduction for the exact version of {\sc X3C}~\cite{GJ79} together with the PCP Theorem~\cite{ALM92,AS92}. 
For a more explicit statement of \Cref{lem:x3c-hardness}, see, e.g., Lemma~27 in the extended version of \cite{GLLM19}.}

\begin{lemma} \label{lem:x3c-hardness}
For some constant $\zeta \in (0, 1/3)$, the following problem is NP-hard: 
Given an {\sc X3C} instance $(u, \calS)$, distinguish between
\begin{itemize}
\item (YES) $\opt_{\setcov}(u, \calS) = u/3$;
\item (NO) $\opt_{\setcov}(u, \calS) \geq u(1/3 + \zeta)$.
\end{itemize}
\end{lemma}

\begin{proof}[of \Cref{thm:balance-hardness}]
Given an instance of {\sc X3C}, we construct an instance of our problem as follows.
First, we create one female candidate for each set $S_i$, and one voter for each element of $[u]$ so that the voter approves all sets $S_i$ to which the element belongs. 
(Hence, each candidate is approved by exactly three voters.)
Next, we create $u/3+1$ additional voters, each of whom approves a new female candidate; this candidate is distinct for distinct voters.
Finally, we create one more voter who approves $2u/3+1$ male candidates.
Let $k = 4u/3+2$, and note that the number of voters is $n = k$.

In the YES case of {\sc X3C}, we can choose $u/3$ original female candidates so that each original voter approves at least one of them.
We then choose all $u/3+1$ new female candidates and all $2u/3+1$ male candidates, and obtain a gender-balanced JR committee.

On the other hand, in the NO case, we must choose at least $u(1/3+\zeta)$ original female candidates in order for every original voter to approve at least one of them.
Moreover, JR requires that we choose all $u/3+1$ new female candidates.
Hence, every JR committee contains at least $(2u/3+1)+\zeta u$ female candidates, and therefore has gender imbalance at least $2\zeta u = 3\zeta(k-2)/2$, which is at least $\zeta k$ for sufficiently large $k$.
Choosing $\varepsilon = \zeta$ yields the desired result.
\hfill $\square$ 
\end{proof}

In spite of this hardness result, \Cref{prop:average-case-half} implies that under the IC model, with high probability, there exists an $n/k$-justifying group of size at most $k/2$.
When this is the case, one can choose the remaining members so as to make the final committee of size $k$ gender-balanced.
In the next section, we show empirically that under several probabilistic models, a small $n/k$-justifying group can usually be found efficiently via the greedy algorithms from Section~\ref{sec:instance-specific}.

\section{Experiments}\label{sec:experiment}
In this section, we conduct experiments to evaluate and complement our theoretical results. 
In the first experiment, we illustrate our probabilistic result for the impartial culture model (Theorem~\ref{thm:average-case}), and examine whether analogous results are likely to hold for two other random models. 
In our second experiment, we analyze how well GreedyCC and GreedyCandidate perform in finding small $n/k$-justifying groups.    
The code for our experiments is available at http://github.com/Project-PRAGMA/Justifying-Groups-SAGT-2022.

\subsection{Set-up} 

We consider three different models for generating approval instances, all of which have been previously studied in the literature \cite{BFKN19}.\footnote{In particular, we refer to the work of Elkind et al.~\cite{EFJS17} for motivation of the Euclidean models.} 
Each model takes as input the parameters $n$ (number of voters), $m$ (number of candidates), and one additional parameter, namely, either an approval probability $p$ or a radius~$r$. 
\begin{itemize}
    \item In the \emph{impartial culture (IC)} model, each voter approves each of the $m$ candidates independently with probability~$p$. 
    This model was already used in Theorem~\ref{thm:average-case}. 
    \item In the \emph{1D-Euclidean (1D)} model, each voter/candidate is assigned a uniformly random point in the interval $[0,1]$. 
    For a voter~$v$ and a candidate~$c$, let $x_v$ and $x_c$ be their respective assigned points. 
    Then, $v$ approves $c$ if and only if $|x_v - x_c| \leq r$. Observe that the resulting profile belongs to the 1D-VCR class discussed in Appendix~\ref{app:restriction}. 
    \item The \emph{2D-Euclidean (2D)} model is a natural generalization of the 1D model where each voter/candidate is assigned a uniformly random point in the unit square $[0,1] \times [0,1]$. 
    Then, a voter $v$ approves a candidate $c$ if and only if the Euclidean distance between their points is at most $r$. 
\end{itemize}

The experiments were carried out on a system with 1.4 GHz Quad-Core Intel Core i5 CPU, 8GB RAM,
and macOS 11.2.3 operating system. 
The software was implemented in Python 3.8.8 and the libraries matplotlib 3.3.4, numpy 1.20.1, and pandas 1.2.4 were used. Additionally, gurobi 9.1.2 was used to solve integer programs.

\begin{figure*}[h!]
\begin{subfigure}{\textwidth} \centering
     \includegraphics[scale =0.21]{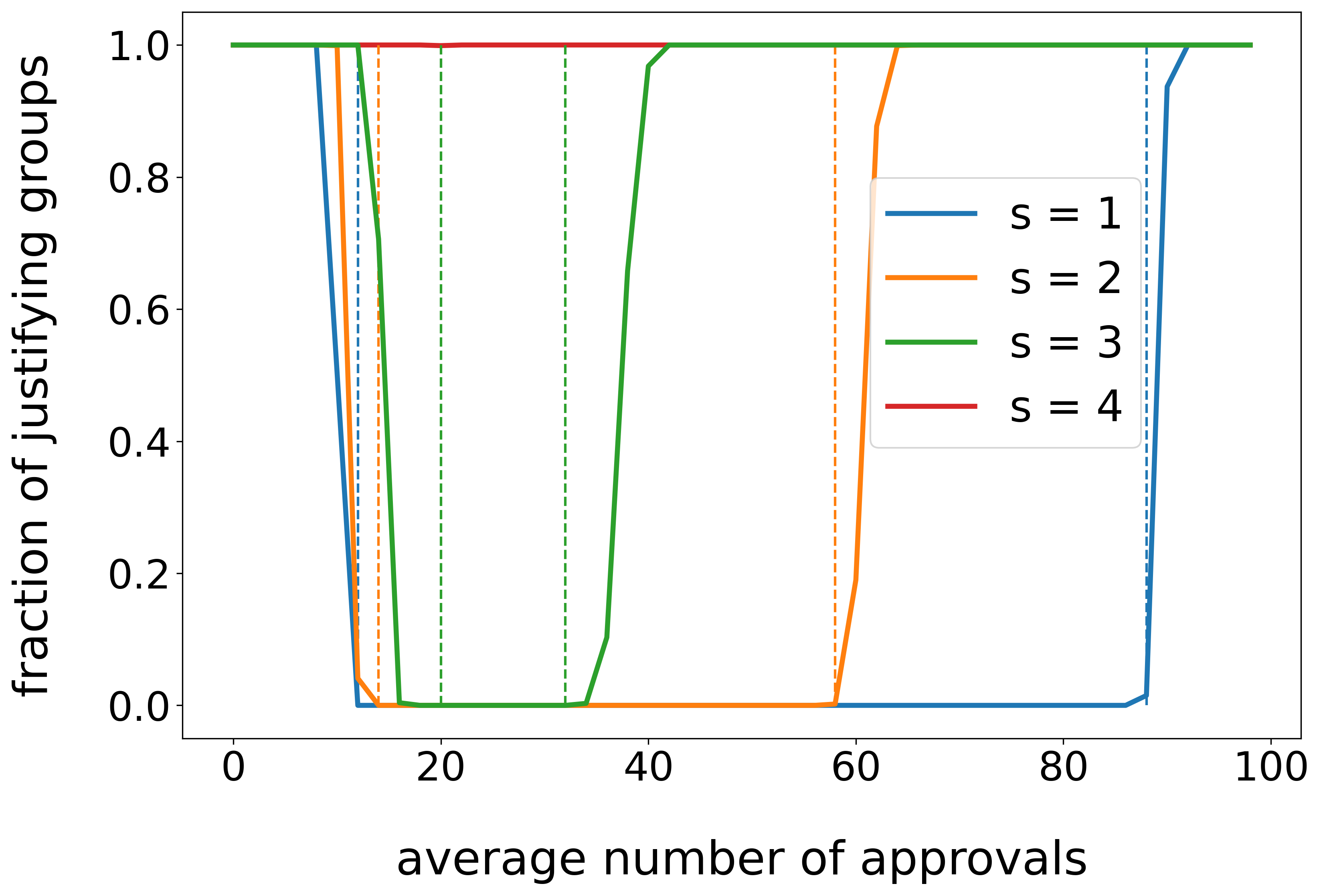}\caption{IC Model} \label{fig:exp1IC}
     \end{subfigure}
     
     \vspace{2mm}
    \begin{subfigure}{0.5\textwidth}\centering
     \includegraphics[scale=0.21]{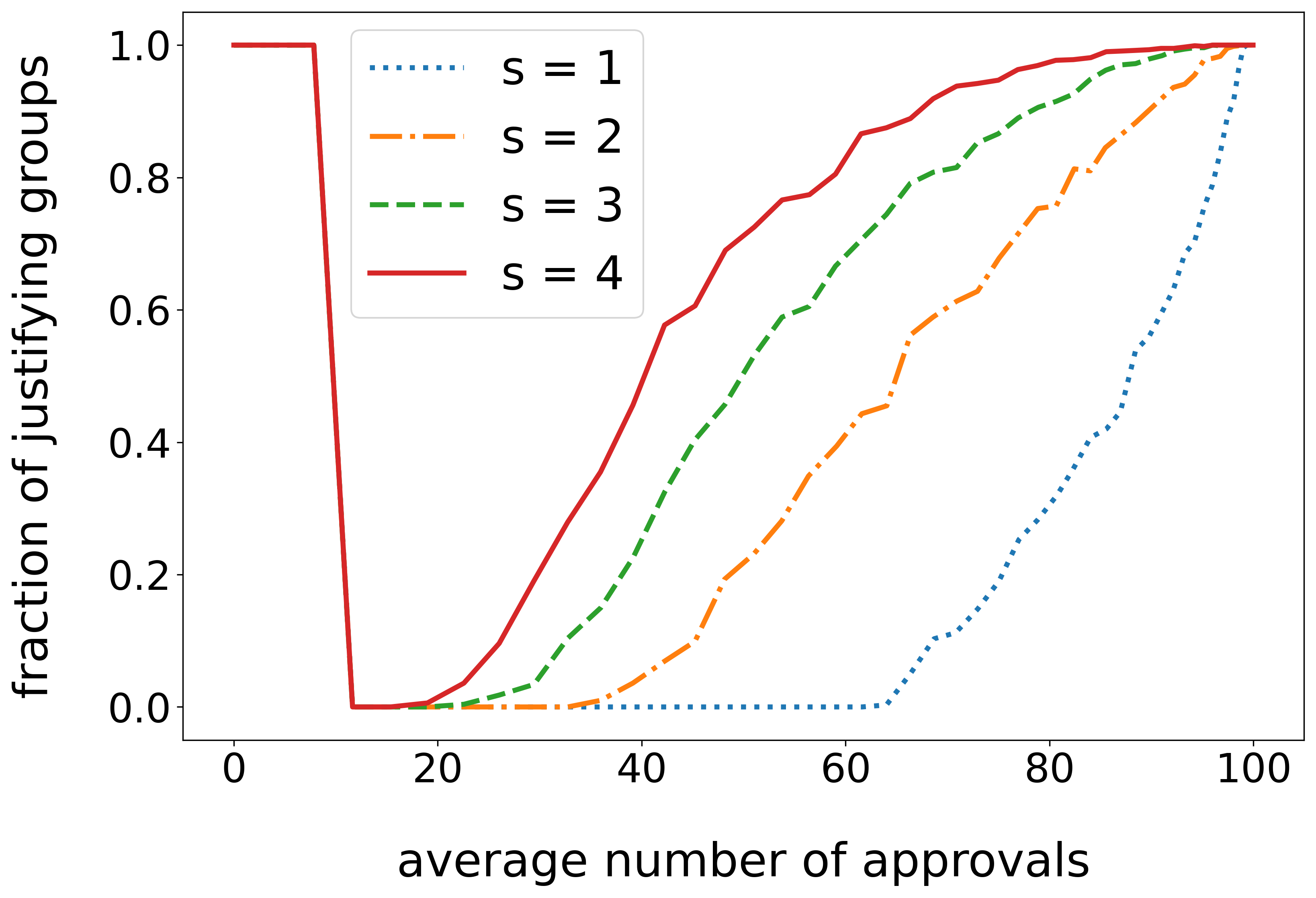}\caption{1D-Euclidean Model}
     \label{fig:exp1_1D}
     \end{subfigure}
     \begin{subfigure}{0.5\textwidth}\centering
     \includegraphics[scale=0.21]{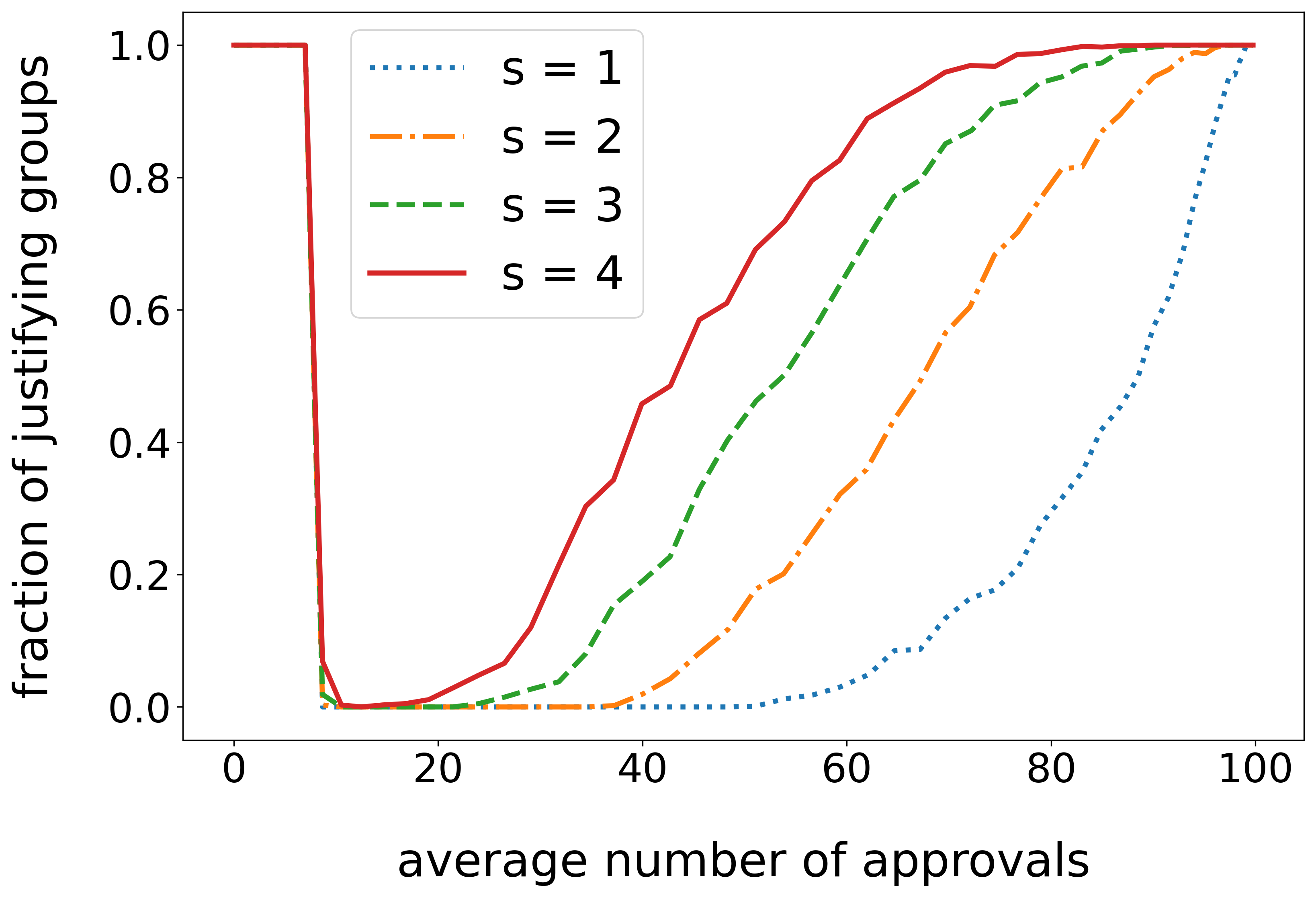}\caption{2D-Euclidean Model}
     \label{fig:exp1_2D}
     \end{subfigure}
    \caption{Experimental results evaluating Theorem~\ref{thm:average-case} as well as analogous settings for two Euclidean models. 
    For each plot, the $x$-axis shows the average number of approvals of a voter for each parameter $p$ (or $r$), and the $y$-axis shows the fraction of the $1000$ generated elections for which a randomly selected size-$s$ group is $n/k$-justifying.
    The dashed vertical lines indicate the transition points for large $n$ as shown in Theorem~\ref{thm:average-case}.}\label{fig:experiment1}
\end{figure*}

\subsection{Empirical Evaluation of Theorem \ref{thm:average-case}}

For our first experiment, we focus on elections with parameters $n=5000$, $m=100$, and $k=10$. 
We chose a large number of voters as the statement of Theorem~\ref{thm:average-case} concerns large values of $n$.  
For each $s \in \{1,2,3,4\}$ and $p \in [0,1)$ (in increments of $0.02$), we generated $1000$ elections using the IC model with parameter $p$. 
We then sampled one group of size $s$ from each resulting election and checked whether it is $n/k$-justifying. 

Figure~\ref{fig:exp1IC} illustrates the fraction of generated elections for which this is the case.
To make this plot comparable to analogous plots for the other two models, we label the $x$-axis with the average number of approvals instead of $p$; this number is simply $p\cdot m = 100p$. 
For each~$s$, the area between the vertical dashed lines indicates the range of the interval $[0,100]$ for which Theorem~\ref{thm:average-case} shows that the probability that no size-$s$ group is $n/k$-justifying converges to $1$ as $n\rightarrow\infty$; this corresponds to the range of $p$ such that $p(1-p)^s > 1/k$. 
For $s=4$, Theorem~\ref{thm:average-case} implies that all size-$s$ groups are likely to be $n/k$-justifying for any average number of approvals in $[0,100]$ (i.e., for any $p \in [0,1]$) as $n\rightarrow\infty$. 
Hence, there are no vertical dashed lines for $s=4$. 

In Figure \ref{fig:exp1IC}, we see that the empirical results match almost exactly the prediction of Theorem \ref{thm:average-case}. 
Specifically, for $s\in \{1,2\}$, we observe a sharp fall and rise in the fraction of $n/k$-justifying groups precisely at the predicted values of $p$.
For $s=3$, the empirical curve falls slightly before and rises slightly after the predicted points marked by the dashed lines. 
This is likely because the function $p(1-p)^3$ is very close to $1/k$ in the transition areas, so $\varepsilon := 1/k - p(1-p)^3$ as defined in the proof of Theorem~\ref{thm:average-case} is very small; thus a larger value of $n$ is needed in order for the transition to be sharp.  

We carried out analogous experiments for the two Euclidean models. 
In particular, we iterated over $s \in \{1,2,3,4\}$ and $r \in [0,1)$ (for the 1D model) or $r \in [0,1.2)$ (for the 2D model), again in increments of $0.02$.
To make the plots for different models comparable, we compute the average number of approvals induced by each value of $r$ and label this number on the $x$-axis.
The resulting plots, shown in Figures~\ref{fig:exp1_1D}--\ref{fig:exp1_2D}, differ significantly from the plot for the IC model. 
In particular, while we see a sharp fall in the fraction of $n/k$-justifying groups when the average number of approvals is around $10$ (for all $s$), there is no sharp rise as in the IC model. 
This suggests that a statement specifying a sharp threshold analogous to Theorem~\ref{thm:average-case} for the IC model is unlikely to hold for either of the Euclidean models. 
Nevertheless, it remains an interesting question whether the fraction of $n/k$-justifying groups can be described theoretically for these models.

\subsection{Performance of GreedyCC and GreedyCandidate}

For our second experiment, we consider elections with parameters $n=m=100$ and $k=10$, and iterate over $p \in [0,1)$ (for the IC model), $r \in [0,1)$ (for the 1D model), and $r \in [0,1.2)$ (for the 2D model), each in increments of $0.02$. 
For each value of $p$ (or $r$), we generated $200$ elections and computed the minimum size of an $n/k$-justifying group (via an integer program) and the size of the $n/k$-justifying group returned by GreedyCC and GreedyCandidate, respectively.
We aggregated these numbers across different elections by computing their average.
As in the first experiment, to make the plots for different models comparable, we converted the values of $p$ and $r$ to the average number of approvals induced by these values. 
The results are shown in Figure~\ref{fig:experiment2}.

\begin{figure*}
\begin{subfigure}{\textwidth} \centering
     \includegraphics[scale =0.21]{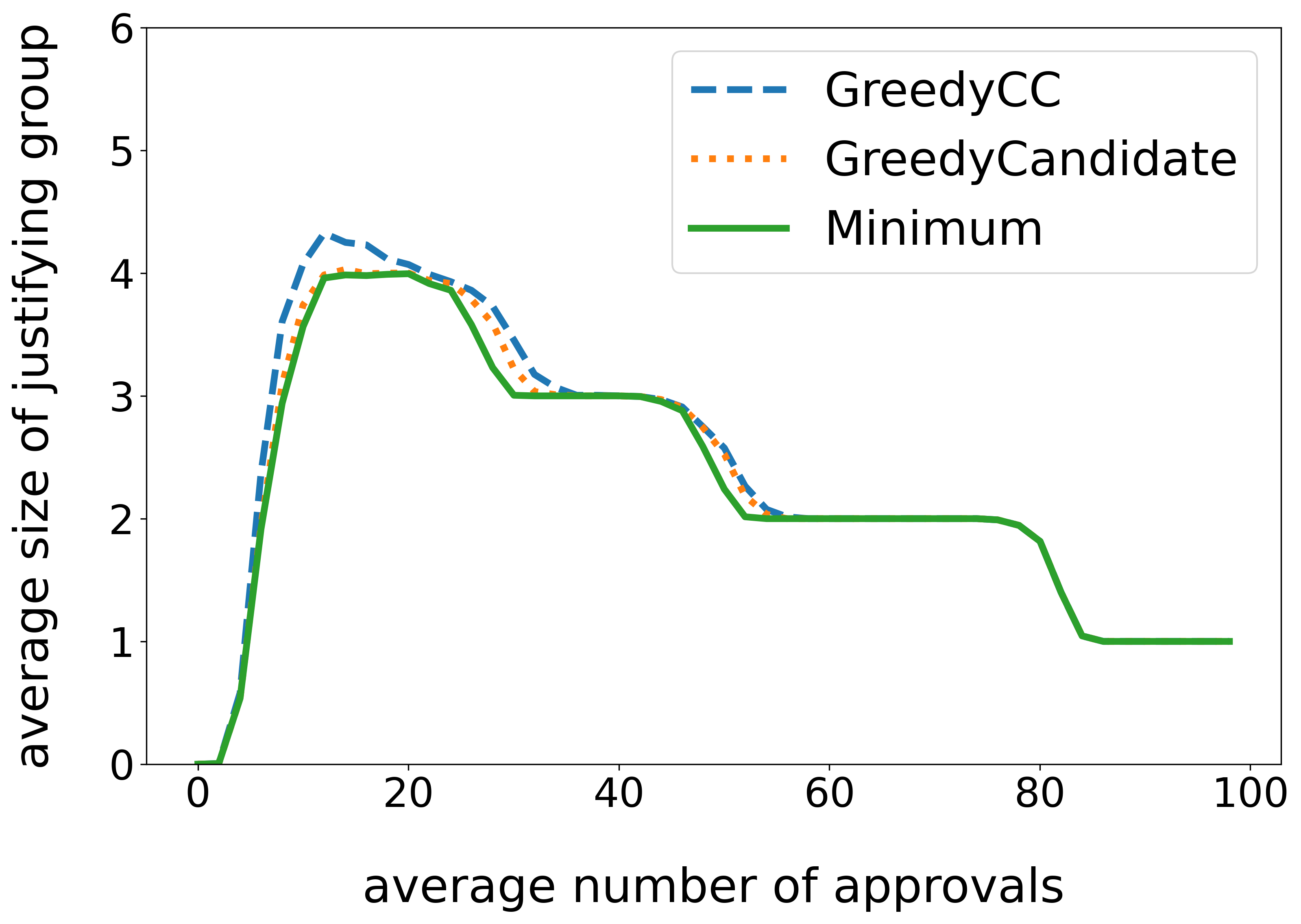}\caption{IC Model } \label{fig:exp2IC}
     \end{subfigure}
     
     \vspace{2mm}
     \begin{subfigure}{0.5\textwidth}\centering
     \includegraphics[scale = 0.21]{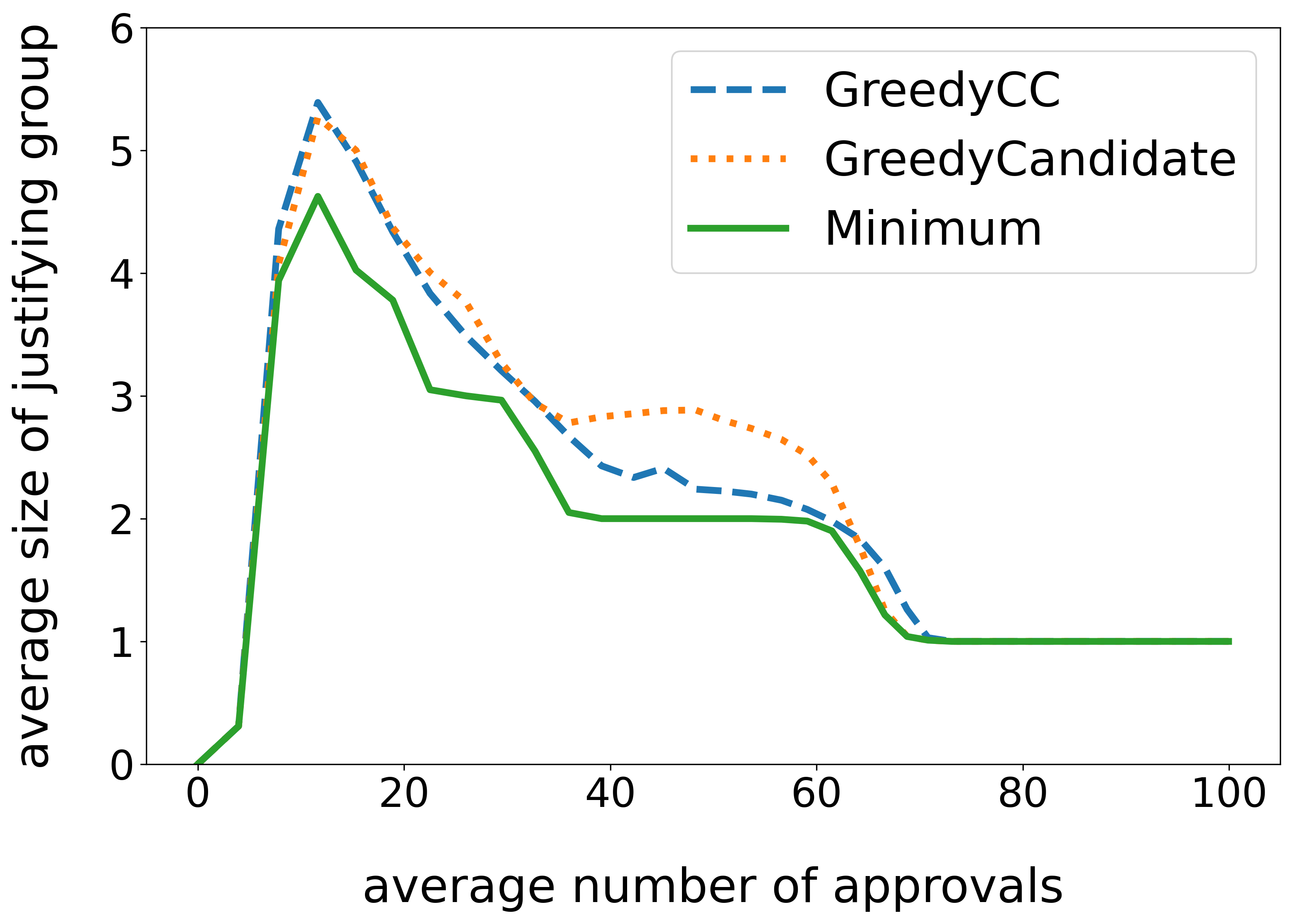}\caption{1D-Euclidean Model}
     \end{subfigure}
     \begin{subfigure}{0.5\textwidth}\centering
     \includegraphics[scale=0.21]{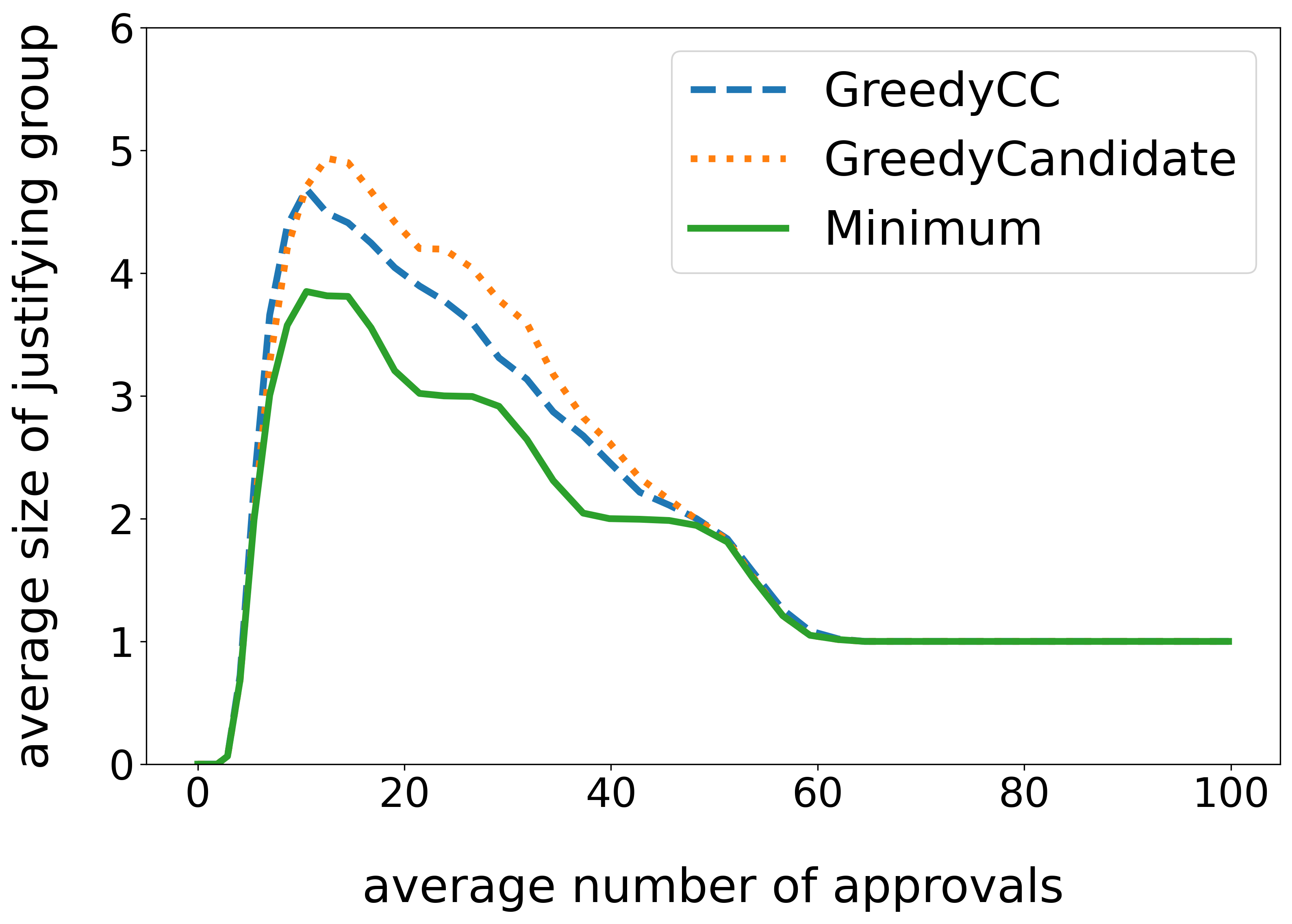}\caption{2D-Euclidean Model}
     \end{subfigure}
    \caption{Experimental results on the performance of GreedyCC and GreedyCandidate. 
    For each plot, the $x$-axis shows the average number of approvals of a voter for each parameter $p$ (or $r$), and the $y$-axis shows the average size of an $n/k$-justifying group output by GreedyCC and GreedyCandidate as well as the average size of a smallest $n/k$-justifying group.
    }\label{fig:experiment2}
\end{figure*}
In general, we observe that both GreedyCC and GreedyCandidate provide decent approximations to the minimum size of an $n/k$-justifying group. 
More precisely, the average difference between the size of a justifying group returned by GreedyCC and the minimum size is less than $1$ for all three models and parameters, while for GreedyCandidate this difference is at most $1.3$. 
The standard deviation of the size of justifying groups returned by GreedyCC and GreedyCandidate is similar for the two Euclidean models and below $1$ for all tested parameters. For the IC model, GreedyCandidate induces a smaller variance than GreedyCC.
Moreover, on average, both greedy algorithms found $n/k$-justifying groups of size at most $k/2 = 5$ for almost all models and parameters---the only exception is the 1D model when the expected number of approvals is around $11$. 
In absolute numbers, for this set of parameters, GreedyCC returned a justifying group of size larger than $k/2$ for $84$ of the $200$ instances and GreedyCandidate for $75$ of the $200$ instances.
Interestingly, among all $32000$ generated instances across all parameters, there was exactly one for which a smallest $n/k$-justifying group was of size larger than $5$.
It is also worth noting that even though GreedyCandidate has a better worst-case guarantee than GreedyCC, this superiority is not reflected in the experiments. 
In particular, while GreedyCandidate performs marginally better than GreedyCC under the IC model, GreedyCC yields slightly better approximations under the Euclidean models.

We also repeated these experiments with $n = 5000$; the results are shown in \Cref{app:experiments}.
Notably, the plot for the IC model shows a clearer step function than Figure~\ref{fig:exp2IC}.

\section{Conclusion and Future Work}

We have investigated the notion of an $n/k$-justifying group introduced by Bredereck et al.~\cite{BFKN19}, which allows us to reason about the justified representation (JR) condition with respect to groups smaller than the target size~$k$.
We showed that $n/k$-justifying groups of size less than $k/2$ typically exist, which means that the number of committees of size $k$ satisfying JR is usually large.
We also presented approximate algorithms for computing a small justifying group as well as an exact algorithm when the instance admits a tree representation.
By starting with such a group, one can efficiently find a committee of size $k$ fulfilling both JR and gender balance, even though the problem is NP-hard in the worst case.

Given the typically large number of JR committees, a natural direction is to impose desirable properties on the committee on top of JR.
In addition to gender balance, several other properties have been studied by Bredereck et al.~\cite{BFI18}.
For instance, when organizing an academic workshop, one could require that at least a certain fraction of the invitees be junior researchers, or that the invitees come from a certain number of countries or continents.
We expect that algorithms for computing small justifying groups will be useful for handling other diversity constraints as well.
It would also be interesting to study analogs of $n/k$-justifying groups for the more demanding representation notions of \emph{proportional justified representation (PJR)} and \emph{extended justified representation (EJR)}, in particular to see whether these analogs yield qualitatively different results.

\subsubsection{Acknowledgments.}

This project has received funding from the European 
Research Council (ERC) under the European Union’s Horizon 2020 
research and innovation programme (grant agreement No 101002854), 
from the Deutsche Forschungsgemeinschaft under grant
BR 4744/2-1, from JST PRESTO under grant number JPMJPR20C1, from the Singapore Ministry of Education under grant number MOE-T2EP20221-0001, and from an NUS Start-up Grant.
We would like to thank the anonymous SAGT reviewers for their comments.
\vspace{0.5mm}
\begin{center}
    \noindent \includegraphics[width=2.5cm]{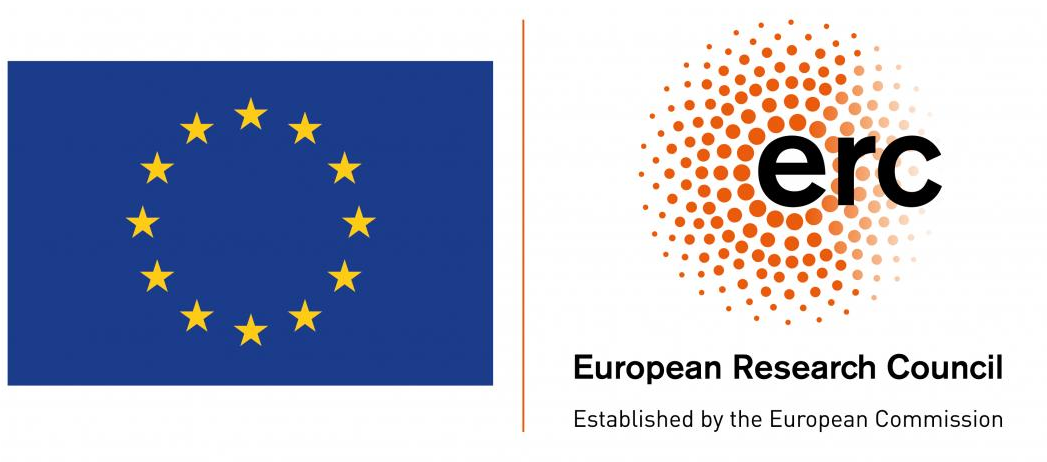}
\end{center}

\bibliographystyle{splncs04}
\bibliography{sagt22}

\appendix

\section{Hardness of Approximating $n/k$-Justifying Group}
\label{app:jr-inapprox}

As we observed in \Cref{sec:contribution}, when $n = k$, the problem of computing a small $n/k$-justifying group for a given instance is equivalent to the {\sc Set Cover} problem, which is NP-hard to approximate to within a factor of $o(\ln n)$.
Below, we show that this hardness holds even when $n > k$.
We first consider the case where the ratio $n/k$ is constant.

\begin{theorem} \label{thm:jr-inapprox}
Let $\ell > 1$ be a constant integer. For any constant $\varepsilon > 0$, it is NP-hard to find an $n/k$-justifying group that is at most $(1 - \varepsilon) \ln n$ times larger than a smallest $n/k$-justifying group, even when $n/k = \ell$.
\end{theorem}

To establish this hardness, we use a reduction from the {\sc Set Cover} problem.
Recall that in {\sc Set Cover}, we are given a universe $[u] = \{1,\dots,u\}$ and a collection $\calS = \{S_1, \dots, S_M\}$ of subsets of $[u]$. 
The goal is to select as few subsets as possible that together cover the universe; we use $\opt_{\setcov}(u, \calS)$ to denote the optimum of a {\sc Set Cover} instance $(u, \calS)$. 

\begin{theorem}[\cite{DinurS14,Moshkovitz15}] \label{thm:setcover-inapprox}
For any constant $\delta > 0$, it is NP-hard to find a set cover of size at most $(1 - \delta) \ln u \cdot \opt_{\setcov}(u, \calS)$.
\end{theorem}

\begin{proof}[of \Cref{thm:jr-inapprox}]
Given an instance of {\sc Set Cover}, we construct an instance of the $n/k$-justifying group problem as follows.
Let $k = u$.
First, we create $n = \ell \cdot u$ voters, with voters $1,2,\dots,u$ corresponding to the {\sc Set Cover} elements.
For every subset $S_i$, we create a candidate $c_{S_i}$ that are approved by exactly the voters in $S_i$. 
In addition, for every $j \in [u]$, we construct a candidate $c^*_j$ that is approved by the $\ell$ voters $j, j + u, \dots, j + (\ell - 1)u$.

To begin with, note that a smallest $n/k$-justifying group of the constructed instance has size at most $\opt_{\setcov}(u, \calS)$, because simply picking the candidates corresponding to the {\sc Set Cover} solution yields an $n/k$-justifying group.

Next, suppose that for some constant $\varepsilon > 0$, there is a polynomial-time algorithm that can find an $n/k$-justifying group of size at most $(1-\varepsilon)\ln n$ times the size of a smallest $n/k$-justifying group; choose $u$ large enough so that $\ln u \ge (2\ln \ell)/\varepsilon$.
Then this algorithm will find an $n/k$-justifying group of size at most $z := (1-\varepsilon)\ln n \cdot \opt_{\setcov}(u, \calS)$. 
Notice that we may assume that this group does not contain any of the candidates $c^*_j$'s, because we may replace a candidate $c^*_j$ by an arbitrary candidate $c_S$ such that $S$ contains $j$. 
Therefore, we may assume that this group consists of candidates $c_{S_{i_1}}, \dots, c_{S_{i_z}}$.
Notice also that $S_{i_1} \cup \dots \cup S_{i_z} = [u]$; this is because, for every $j \in [u]$, voters $j, j + u, \dots, j + (\ell - 1)u$ approve a common candidate $c^*_j$ but $j + u, \dots, j + (\ell - 1)u$ do not approve any of the selected candidates. 
As a result, we have found a set cover of size $z$ in polynomial time.
Note that 
\begin{align*}
z &=  (1-\varepsilon)\ln n\cdot \opt_{\setcov}(u, \calS) \\
&= (1 - \varepsilon)(\ln u + \ln\ell)\cdot \opt_{\setcov}(u, \calS) \\
&\le (1 - \varepsilon)\left(\ln u + \frac{\varepsilon\ln u}{2}\right)\cdot \opt_{\setcov}(u, \calS) \\
&\le \left(1 - \frac{\varepsilon}{2}\right)\ln u\cdot \opt_{\setcov}(u, \calS),
\end{align*}
which means that the set cover that we have found in polynomial time has size at most $(1-\delta)\ln u\cdot \opt_{\setcov}(u, \calS)$, where $\delta := \varepsilon/2$.
But from~\Cref{thm:setcover-inapprox}, this is NP-hard.
\hfill $\square$ 
\end{proof}

What happens if $n$ is larger than $k$ by a superconstant factor?
Even in that case, the reduction in \Cref{thm:jr-inapprox} can still be modified to yield a similar hardness result.

\begin{theorem}
Let $d > 1$ be a constant integer.
For any constant $\varepsilon > 0$, it is NP-hard to find an $n/k$-justifying group that is at most $\frac{1}{d}(1-\varepsilon)\ln n$ times larger than a smallest $n/k$-justifying group, even when $n = k^d$.
\end{theorem}

\begin{proof}
We use the same construction as in the proof of \Cref{thm:jr-inapprox}, choosing $k = u$ and $\ell = n/k = k^{d-1}$.
If for some constant $\varepsilon > 0$ there is a polynomial-time algorithm that can find an $n/k$-justifying group of size at most $\frac{1}{d}(1-\varepsilon)\ln n$ times the size of a smallest $n/k$-justifying group, then a similar argument as in the proof of \Cref{thm:jr-inapprox} shows that we can use this algorithm to find a set cover of size at most $\frac{1}{d}(1-\varepsilon)\ln n\cdot \opt_{\setcov}(u, \calS)
= (1-\varepsilon)\ln u\cdot \opt_{\setcov}(u, \calS)$ in polynomial time.
But from~\Cref{thm:setcover-inapprox}, this is NP-hard.
\hfill $\square$ 
\end{proof}

\section{Preference Restrictions}
\label{app:restriction}

We have shown in \Cref{sec:tree-representation} that the problem of computing a smallest $n/k$-justifying group can be solved efficiently for instances that admit a tree representation (TR).
In order to better understand the class TR and related preference restriction classes, we explore the relationships between them in this section.

First, we exhibit that TR contains a recently introduced class called \emph{1D voter/candidate range model (1D-VCR)} \cite{GBSF21}.
For convenience, we state the definition of 1D-VCR here.
For a 1D-VCR instance $I=(C, \calA, k)$, each $a \in C \cup N$ has a center of influence $x_a$ and a radius of influence $r_a$.
We denote by $s(a):=x_a-r_a$ and $t(a):=x_a+r_a$ the leftmost and the rightmost point of $a$'s range of influence, respectively.
We call $J_a := [s(a),t(a)]$ the \emph{interval} of $a$. 
A voter $i\in N$ approves a candidate $c\in C$ if and only if $c$'s interval $[s(c), t(c)]$ and $i$'s interval $[s(i), t(i)]$ have a non-empty intersection, i.e.,
\begin{align*}
s(c)  \le t(i) ~\mbox{and}~ s(i) \le t(c).
\end{align*}

Before establishing the containment, we prove some basic properties of 1D-VCR
preferences.  Our first observation is that if a candidate $c'$ is
more appealing than candidate $c$, in the sense that the interval of
$c$ is contained in that of $c'$, then every voter who
approves $c$ also approves~$c'$.

\begin{lemma}\label{lem0}
Consider a 1D-VCR instance. If a voter $i$ approves candidate $c$, then $i$ approves any candidate $c'$ whose interval contains that of $c$, i.e., $J_c \subseteq J_{c'}$.
\end{lemma}

\begin{proof}
Since voter $i$ approves candidate $c$, we have $J_i\cap J_c\neq\emptyset$.
Since $J_c\subseteq J_{c'}$, it holds that $J_i\cap J_{c'}\neq\emptyset$.
We conclude that $i$ approves $c'$.
\hfill $\square$ 
\end{proof}

An interval $[x,y]$ is said to be \emph{nested} in
another interval $[x',y']$ if $x' \le x$ and $y \le y'$, with at least
one inequality being strict.  
The next lemma ensures that if a voter $i$ approves two candidates whose intervals are not nested in
each other's, then $i$ approves any ``intermediate'' candidate whose
interval lies between the intervals of the two approved candidates.

\begin{lemma}\label{lem:ordering}
Consider a 1D-VCR instance. If a voter $i$ approves candidates $a$ and $b$ with $s(a) \le s(b)$ and $t(a) \le t(b)$, then $i$ also approves any candidate $c$ such that $s(a) \le s(c) \le s(b)$ and $t(a) \le t(c) \le t(b)$.
\end{lemma}

\begin{proof}
Since voter $i$ approves candidates $a$ and $b$, we have $s(i) \le t(a)$ and $s(b) \le t(i)$.
Hence, $s(c) \le s(b) \le t(i)$ and $s(i) \le t(a) \le t(c)$.
This means that $i$ approves $c$, as claimed.
\hfill $\square$ 
\end{proof}

We are now ready to show the containment relation between 1D-VCR and TR.

\begin{proposition}
Every 1D-VCR instance admits a TR. 
Moreover, such a TR can be computed in polynomial time.
\end{proposition}

\begin{proof}
Let $I = (C,\mathcal{A},k)$ be a 1D-VCR instance with voter set $N$. 
We construct a tree $T$ as follows. 
Consider a maximal unnested subset of $C$---call it $C_0$---where a candidate is said to be \emph{nested} if its interval is nested in another candidate's interval.
We reindex the candidates in $C_0$ so that $s(c_1) \leq \dots \leq s(c_{\ell})$ and $t(c_1) \leq \dots \leq t(c_{\ell})$, where $\ell := |C_0|$.
Then, we add the path $(c_1, \dots, c_{\ell})$ to $T$, call these nodes the ``level-$0$ nodes'', and define $C' := C \setminus C_0$. 
For the remaining candidates in $C'$, we iteratively apply the following procedure: pick a candidate $c \in C'$ whose interval is not nested in the interval of any other candidate in $C'$, and make $c$ a child of a node in $T$ whose interval strictly contains $J_c$ (such a node exists by definition of $C_0$), breaking ties in favor of nodes with a higher level.
Remove $c$ from $C'$, and define the level of $c$ as the level of its parent plus $1$. 

We claim that the following two statements hold:
\begin{enumerate}[label=(\roman*)]
\item Let $c \in A_i$ for some $i \in N$, and $c' \in C_0$ be the level-$0$ ancestor of $c'$ (possibly $c=c'$). 
Then, all candidates on the path from $c$ to $c'$ are in $A_i$. 
\item Let $\{c_p,c_q\} \subseteq A_i \cap C_0$ for some $i \in N$ and $1\le p<q\le \ell$. 
Then, $\{c_p,c_{p+1}, \dots, c_q\} \subseteq A_i$.
\end{enumerate}

To prove (i), let $d_0,d_1,\dots,d_{r}$ be the path from $c'$ to $c$; in particular, $c'=d_0$ and $c=d_r$. 
By construction of $T$, it holds that $J_{d_r} \subseteq J_{d_{r-1}} \subseteq \dots \subseteq J_{d_0}$. 
Since $d_r\in A_i$, Lemma~\ref{lem0} implies that $\{d_0,d_1,\dots,d_r\} \subseteq A_i$. 
As for (ii), since $s(c_p) \leq s(c_{p+1}) \leq \dots \leq s(c_{q})$ and $t(c_p) \leq t(c_{p+1}) \leq \dots \leq t(c_{q})$, Lemma~\ref{lem:ordering} together with the assumption that $\{c_p,c_q\}\subseteq A_i$ imply that $\{c_p,c_{p+1}, \dots, c_q\} \subseteq A_i$. 

Clearly, $T$ can be constructed in polynomial time; we now show that it is a valid tree representation of the instance $I$. 
To this end, fix $c,c' \in A_i$ for some voter $i \in N$. 
It suffices to show that there is a walk from $c$ to $c'$, possibly going through some nodes more than once, such that all candidates on this walk belong to $A_i$.
Consider a walk composed of the path from $c$ to its level-$0$ ancestor $d$, the path from $d$ to the level-$0$ ancestor $d'$ of $c'$, and the path from $d'$ to $c'$.
By (i), all nodes in the first and third paths are in $A_i$; by (ii), all nodes in the second path are in $A_i$ too.
This means that the entire walk is contained in $A_i$, completing the proof.
\hfill $\square$ 
\end{proof}

In the chain of tree representation classes depicted by Yang~\cite[Fig.~4]{yang2019tree}, the largest class contained in TR is the class of PTR.
An instance admits a \emph{path-tree representation (PTR)} if there exists a tree $T$ with vertex set corresponding to the candidate set $C$ such that the approval set of every voter induces a path in~$T$.
Below, we present examples demonstrating that PTR and 1D-VCR do not contain each other---this further highlights the generality of TR and also shows that the inclusion of 1D-VCR in TR is strict.

\begin{proposition}
There exists a 1D-VCR instance that admits no PTR. 
\end{proposition}

\begin{proof}
Consider the following instance: $C = \{a,b,c,d\}$,  $n=4$, and
\begin{align*}
A_1 &= \{a,b,c,d\}, 
A_2 = \{a,b\}, A_3 = \{a,c\}, A_4=\{a,d\}.
\end{align*}
To see that this instance admits a 1D-VCR representation, consider the following intervals: 
\begin{align*}
J_a &= [0,5], J_b = [0,1],J_c = [2,3],J_d=[4,5]; \\ 
J_1 &= [0,5],J_2=[0,1],J_3=[2,3],J_4=[4,5].
\end{align*}
Now, assume for contradiction that the instance admits a path-tree representation. 
Then, since $A_1=\{a,b,c,d\}$, all four candidates must lie on a path in the tree. 
But then $a$ has at most two neighbors, so one of the approval sets $A_2$, $A_3$, and $A_4$ is not a path, a contradiction.
\hfill $\square$ 
\end{proof}

\begin{proposition}
There exists an instance that is not 1D-VCR but admits a PTR.
\end{proposition}

\begin{proof}
Consider the following instance: $C=\{a,b,c,d\}$, $n=3$, and 
\[
A_1 = \{a,b,c\}, A_2 = \{a,b,d\}, A_3 = \{a,c,d\}.
\]
This instance admits a PTR representation in which $a$ is the center of a star graph with three leaves. 
Now, assume for contradiction that the instance admits a 1D-VCR representation, and let 
$J_b=[s(b),t(b)]$, $J_c=[s(c),t(c)]$, and $J_d=[s(d),t(d)]$
be the intervals of the candidates $b$, $c$, and $d$.
For every pair of candidates $z,z'\in\{b,c,d\}$, there exists a voter who approves $z$ but not $z'$, so none of the intervals $J_b$, $J_c$, $J_d$ can be nested in another one.
Hence, we may assume without loss of generality that $s(b) \leq s(c) \leq s(d)$ and $t(b) \leq t(c) \leq t(d)$. 
Let $J_2 = [s(2), t(2)]$ be the interval of voter~$2$. 
As $c\not\in A_2$, we have $J_2\cap J_c=\emptyset$, so either
$t(2)<s(c)$ or $s(2)>t(c)$.
In the former case, $J_2\cap J_d=\emptyset$, contradicting $d\in A_2$, while in the latter case, $J_2\cap J_b=\emptyset$, contradicting $b\in A_2$.
\hfill $\square$ 
\end{proof}

\section{Additional Experiments}
\label{app:experiments}

\begin{figure*}[t!]
\begin{subfigure}{\textwidth} \centering
     \includegraphics[scale =0.21]{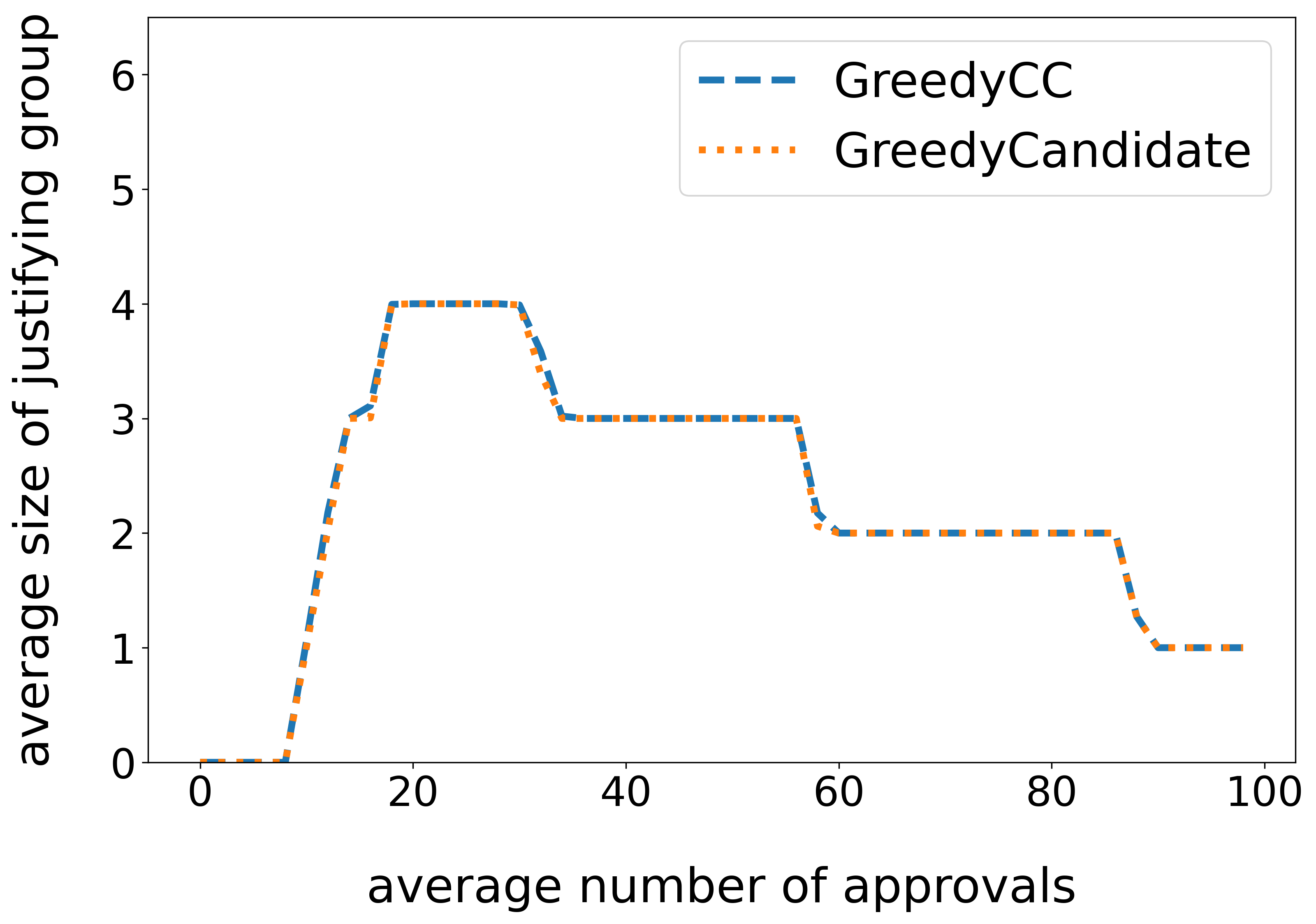}\caption{IC Model } \label{fig:exp3IC}
     \end{subfigure}
     
     \vspace{2mm}
     \begin{subfigure}{0.5\textwidth}\centering
     \includegraphics[scale = 0.21]{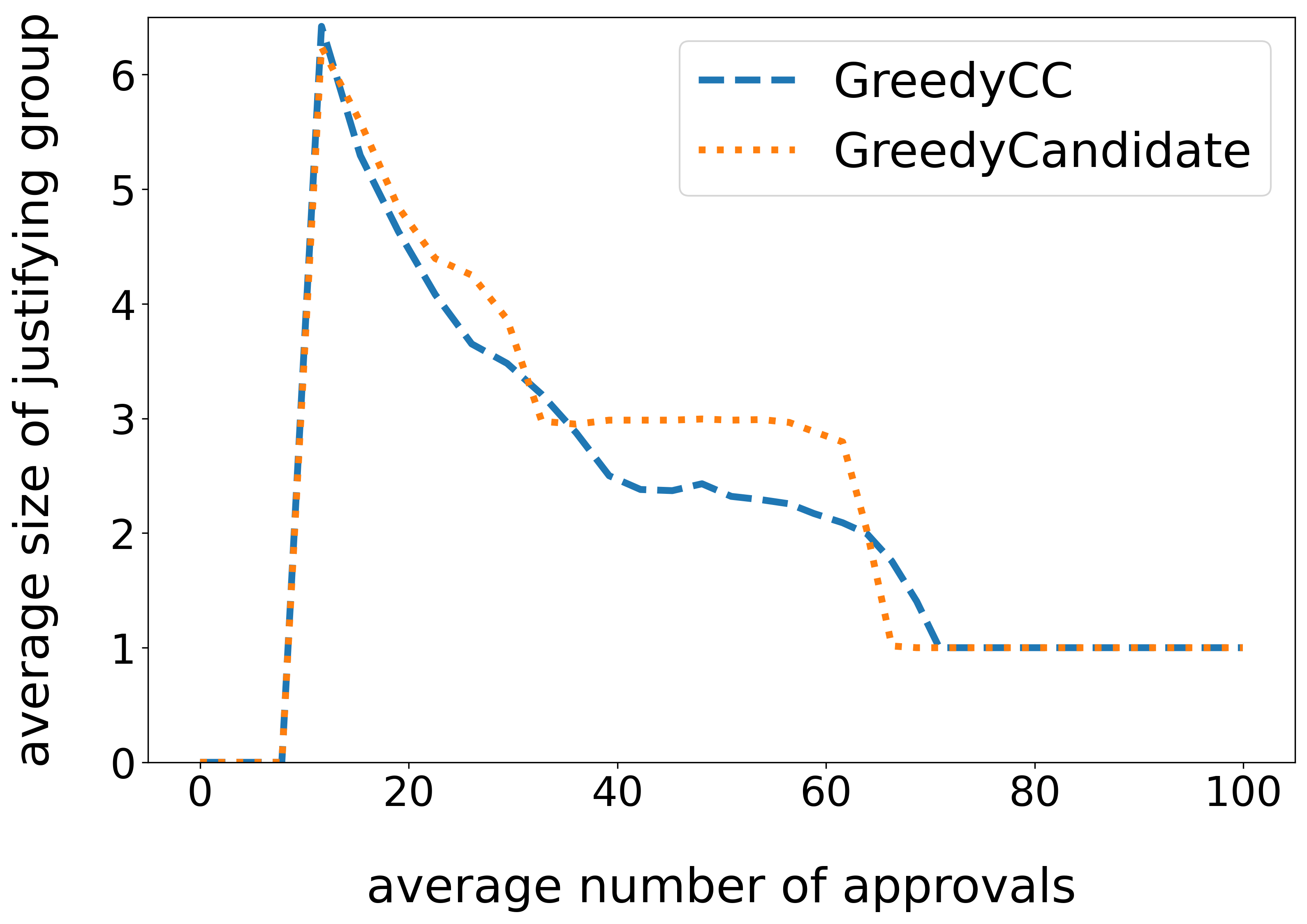}\caption{1D-Euclidean Model}
     \end{subfigure}
     \begin{subfigure}{0.5\textwidth}\centering
     \includegraphics[scale=0.21]{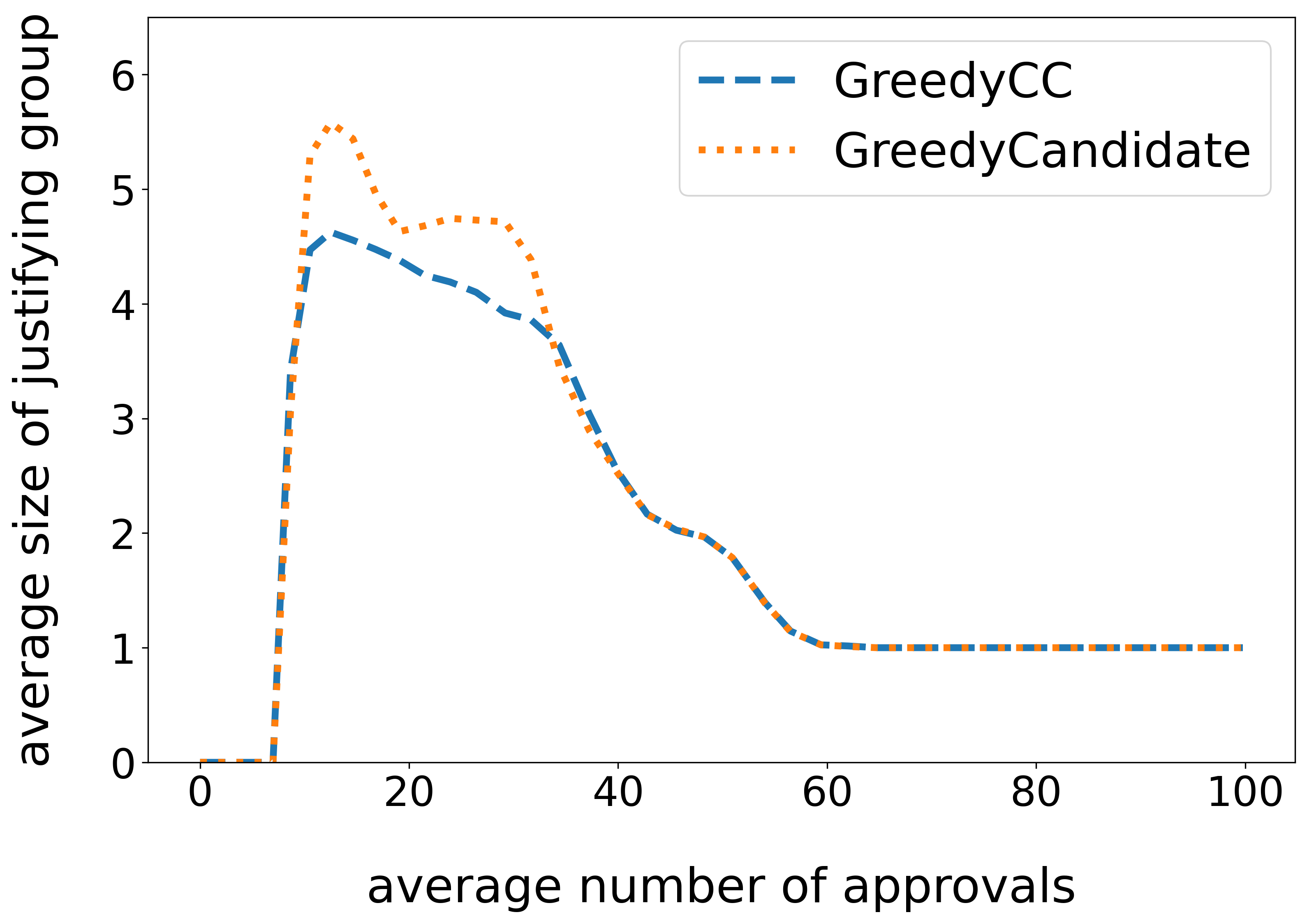}\caption{2D-Euclidean Model}
     \end{subfigure}
    \caption{Experimental results on the performance of GreedyCC and GreedyCandidate with $n=5000$. 
    For each plot, the $x$-axis shows the average number of approvals of a voter for each parameter $p$ (or $r$), and the $y$-axis shows the average size of an $n/k$-justifying group output by GreedyCC and GreedyCandidate.
    }\label{fig:experiment3}
\end{figure*}

We repeated our experiments from \Cref{fig:experiment2} with an increased number of voters: we set $n=5000$ while keeping the remaining parameters unchanged. 
More precisely, we created elections with parameters $m=100$ and $k=10$, and iterated over $p \in [0,1)$ (for the IC model), $r \in [0,1)$ (for the 1D model), and $r \in [0,1.2)$ (for the 2D model), each in increments of $0.02$. 
For each value of $p$ (or $r$), we generated $200$ elections and computed the size of the $n/k$-justifying group returned by GreedyCC and GreedyCandidate, respectively (unfortunately, due to the high number of voters, computing the size of a smallest $n/k$-justifying group was infeasible). 
We aggregated these numbers across different elections by computing their average.
As in the previous experiments, to make the plots for different models comparable, we converted the values of $p$ and $r$ to the average number of approvals induced by these values. 
The results are shown in Figure~\ref{fig:experiment3}.

The plot for the IC model shows a clear step function. 
Considering the corresponding plot in \Cref{fig:experiment2}, it appears likely that this function also represents the size of a smallest $n/k$-justifying group. 
It is worth noting that for large values of $n$, \Cref{thm:average-case} suggests that the size of a smallest $n/k$-justifying group in the IC model can be predicted from the parameters $p$ and $k$: specifically, it is $\tau(p,k) := \lceil-\log_{1-p}(kp)\rceil$.
If all groups of size $\tau(p,k)$ are $n/k$-justifying while all smaller groups are not, then both GreedyCC and GreedyCandidate return a group of this size.
A closer look at our data indicates that this behavior occurs for most values of $p$, as the standard deviation of the size of the returned group is extremely small. 
In particular, it is $0$ for almost all values of $p$ with both algorithms, and below $0.5$ for all parameters. 

By contrast, for the two Euclidean models, the plots are relatively far from step functions. 
It is also unclear how the corresponding plots for the size of a smallest $n/k$-justifying group would look like.
Moreover, the standard deviation of the size of the group returned by GreedyCC and GreedyCandidate is significantly larger than in the IC model. 
Specifically, the standard deviation is nonzero for a large fraction of values of $r$, and can be as high as $0.9$.

\end{document}